\documentclass{amsart}
\usepackage{amssymb}
\usepackage{mathtools}    
\usepackage{color}
\usepackage[hidelinks]{hyperref}
\usepackage{tikz}

\usepackage[mathscr]{eucal}

\newtheorem{thm}{Theorem}
\newtheorem{prop}[thm]{Proposition}
\newtheorem{lem}[thm]{Lemma}
\newtheorem{cor}[thm]{Corollary}
\theoremstyle{definition}
\newtheorem{rem}[thm]{Remark}

\numberwithin{equation}{section}
\numberwithin{thm}{section}

\newcommand{\NN}{\mathbb{N}}
\newcommand{\ZZ}{\mathbb{Z}}
\newcommand{\CC}{\mathbb{C}}
\newcommand{\RR}{\mathbb{R}}
\newcommand{\FF}{\mathbb{F}}
\renewcommand{\leq}{\leqslant}
\renewcommand{\geq}{\geqslant}
\newcommand{\ii}{\mathrm{i}}

\newcommand{\disc}{\mathscr{D}}
\newcommand{\app}[1]{\tilde{#1}}
\newcommand{\appdisc}{\app{\disc}}
\newcommand{\error}{\varepsilon}
\newcommand{\divides}{\mathrel{|}}

\DeclareMathOperator{\Mcost}{\mathsf{M}}
\DeclareMathOperator{\Mcostinc}{\overline{\mathsf{M}}}
\DeclareMathOperator{\Tcost}{\mathsf{T}}

\newcommand{\step}[1]{\medskip\textit{#1}}

\begin{document}

\author{\sc David Harvey}
\address{School of Mathematics and Statistics, University of New South Wales, Sydney NSW 2052, Australia}
\curraddr{}
\email{d.harvey@unsw.edu.au}
\thanks{David Harvey was supported by ARC Future Fellowship grant FT160100219.}

\author{\sc Joris van der Hoeven}
\address{CNRS, Laboratoire d'informatique, \'Ecole polytechnique, 91128 Palaiseau, France}
\curraddr{}
\email{vdhoeven@lix.polytechnique.fr}
\thanks{Joris van der Hoeven was supported by an ERC-2023-ADG grant for the
ODELIX project (number 101142171).}


\date{}
\dedicatory{}

\title[Multiplication is at least as hard as transposition]%
   {Integer multiplication is at least as hard as matrix transposition}

\begin{abstract}
   Working in the multitape Turing model,
   we show how to reduce the problem of matrix transposition
   to the problem of integer multiplication.
   If transposing an $n \times n$ binary matrix
   requires $\Omega(n^2 \log n)$ steps on a Turing machine,
   then our reduction implies that multiplying $n$-bit integers
   requires $\Omega(n \log n)$ steps.
   In other words, if matrix transposition is as hard as expected,
   then integer multiplication is also as hard as expected.
\end{abstract}

\maketitle

\section{Introduction}

We work throughout in the \emph{multitape Turing model}.
In this model, an ``algorithm'' is a Turing machine
with a fixed, finite number of one-dimensional tapes,
and the time complexity of an algorithm refers to the number of steps
executed by the machine.
For detailed definitions see for instance
\cite[\S1.6]{AHU-algorithms} or \cite[Ch.~2]{Pap-complexity}.

For any integer $n \geq 1$ we define
\[
   \lg n \coloneqq \max(\lceil \log_2 n \rceil, 1).
\]
For an integer $\ell \geq 0$ we write $\lg^{\circ \ell}$ for the
$\ell$-fold composition of $\lg$.
Since $1 \leq \lg n < n$ for all $n \geq 2$,
it makes sense to define the \emph{iterated logarithm} function,
\[
   \lg^* n \coloneqq \min \{ \ell \geq 0: \lg^{\circ \ell} n = 1\},
      \qquad n \geq 1.
\]

For functions $f(x)$ and $g(x)$ defined on some domain $D \subseteq \ZZ^n$,
the statement $f(x) = O(g(x))$ (respectively $f(x) = \Omega(g(x))$)
means that there exists a constant $C > 0$
such that $f(x) \leq C g(x)$ (respectively $f(x) \geq C g(x)$)
for all $x \in D$.
For functions $f(n)$ and $g(n)$ defined on $\ZZ^+$,
we write $f(n) = \omega(g(n))$ to mean that
$\lim_{n \to \infty} f(n)/g(n) = \infty$.

\subsection{Multiplication and transposition}
\label{sec:mult-and-tp}

In this paper we study the relationship between the complexity
of two fundamental problems in the Turing model:
\begin{itemize}
   \item
   \emph{Integer multiplication}.
   The input consists of an integer $m \geq 1$
   and two non-negative $m$-bit integers $x$ and $y$.
   The output is the product $xy$.
   All integers are assumed to be stored in the
   standard binary representation.

   A machine $M$ computing this function will be called a
   \emph{multiplication machine}.
   We write $\Mcost_M(m)$ for the worst-case running time of $M$
   on $m$-bit inputs.
   The subscript $M$ may be omitted if it is clear from context.
   
   \item
   \emph{Matrix transposition}.
   The input consists of integers $n_1, n_2, b, m \geq 1$
   such that $n_1 n_2 b \leq m$,
   and an $n_1 \times n_2$ matrix with $b$-bit entries.
   The matrix is encoded on the tape in row-major order,
   i.e., as an array of $b$-bit strings of length $n_1 n_2$,
   with the matrix entry $A_{i_1,i_2}$
   stored at index $n_2 i_1 + i_2$ in the array,
   for $0 \leq i_1 < n_1$ and $0 \leq i_2 < n_2$.
   The output consists of the transpose of this matrix,
   i.e., the matrix entries must be rearranged
   into an $n_2 \times n_1$ array,
   with $A_{i_1,i_2} = (A^T)_{i_2,i_1}$ now appearing
   at position $n_1 i_2 + i_1$.
   Alternatively, one may think of this operation as switching from
   row-major to column-major representation.
   
   (The parameter $m$ is included for technical reasons:
   our matrix transposition algorithms will often need to ``know''
   what size integer multiplication problem is being targeted.
   One may think of $m$ as an upper bound for the total bit size
   of the matrix, which is $n_1 n_2 b$.)
   
   A machine $T$ computing this function will be called a
   \emph{transposition machine}.
   We write $\Tcost_T(m; n_1, n_2, b)$ for the worst-case running time
   of $T$ for the given parameters.
\end{itemize}

The case of \emph{square binary matrices},
where $n_1 = n_2$ and $b = 1$,
is of particular interest.
A machine handling this case will be called a
\emph{binary transposition machine}
(for brevity we omit ``square'').
It is convenient to describe the input size by the single parameter $m$,
so we assume that such a machine takes as input an integer $m \geq 1$
and an $n \times n$ binary matrix,
where $n \coloneqq \lfloor m^{1/2} \rfloor$.
We denote its worst-case running time by $\Tcost_T(m)$.

The following upper bounds are known for the complexity of these problems.
For integer multiplication,
the present authors
recently described a (somewhat complicated)
algorithm $M^*$ \cite{HvdH-nlogn} that achieves
\begin{equation}
   \label{eq:Mstar-bound}
   \Mcost_{M^*}(m) = O(m \lg m).
\end{equation}
For matrix transposition,
there is a simple folklore algorithm $T^*$ achieving
\begin{equation}
   \label{eq:Tstar-bound}
   \Tcost_{T^*}(m; n_1, n_2, b)
      = O(n_1 n_2 b \lg \min(n_1, n_2))
      = O(m \lg \min(n_1, n_2)).
\end{equation}
In particular, for the binary case this algorithm achieves
\[
   \Tcost(m) = O(m \lg m).
\]
Briefly, the algorithm operates as follows.
If $n_1 \leq n_2$,
it first recursively transposes the ``top'' and ``bottom''
halves of the matrix,
and then interleaves the results together appropriately.
The $n_1 > n_2$ case may be handled by running the same
algorithm in reverse.
(We do not know the precise history of this algorithm.
The basic idea goes back at least
to~\cite{Stone-shuffle,Floyd-permuting,Paul-transpose},
although these papers work in somewhat different complexity models.
For a more modern presentation, see for example \cite[Lem.~18]{BGS-recurrences}.)

What is believed to be the true complexity of integer multiplication
and matrix transposition?
For both problems,
non-trivial lower bounds are known
only in certain very restricted models of computation
(see Section \ref{sec:known-lower-bounds}).
If we allow the full power of the multitape Turing machine,
then unfortunately no lower bounds are known beyond the trivial linear bounds
$\Mcost(m) = \Omega(m)$ and $\Tcost(m; n_1, n_2, b) = \Omega(n_1 n_2 b)$.
(These lower bounds must hold because the machine needs enough time
to write its output.)
It seems reasonable to guess that for both problems,
the upper bounds mentioned above are sharp,
i.e., we suspect that for any multiplication machine $M$,
\begin{equation}
   \label{eq:mul-lower-bound}
   \Mcost_M(m) = \Omega(m \lg m),
\end{equation}
and for any transposition machine $T$,
\[
   \Tcost_T(m; n_1, n_2, b) = \Omega(n_1 n_2 b \lg \min(n_1, n_2)).
\]
In particular, for any binary transposition machine $T$, we expect that
\[
   \Tcost_T(m) = \Omega(m \lg m).
\]
The lower bound \eqref{eq:mul-lower-bound}
was already suggested by Sch\"onhage and Strassen
in their famous paper \cite{SS-multiply}.
We are not aware of any officially conjectured lower bound for transposition
in the literature, although the absence of non-trivial lower bounds
in the Turing model for transposition and other basic operations
has attracted some attention \cite{Reg-superlinear}.

\begin{rem}
   In a statement of the form \eqref{eq:mul-lower-bound},
   the implied big-$\Omega$ constant depends on the machine $M$.
   This is unavoidable, because the \emph{linear speedup theorem}
   \cite[Thm.~2.2]{Pap-complexity}
   states that for any Turing machine $A$ running in time $f(m)$
   on input of size $m$,
   and for any constant $C > 1$,
   we can construct another Turing machine $A'$ that solves the same problem
   as $A$ in time at most $f(m)/C + O(m)$.
   (The new machine simulates the original machine,
   executing several steps of the original machine in each step,
   by using a packed encoding of the original alphabet.)
\end{rem}

\subsection{Summary of main results}
\label{sec:summary}

The aim of this paper is to present various reductions from
matrix transposition to integer multiplication.
We state here some sample consequences of these reductions,
concentrating on the binary case for simplicity.
The bottom line is that
\emph{lower bounds for matrix transposition imply
lower bounds for integer multiplication}.
The proofs of all results in this section
(and some more general results)
will be given in Section \ref{sec:proofs}.

First, we have the following striking relationship between
the lower bound conjectures for multiplication and transposition
mentioned earlier.
\begin{thm}
   \label{thm:main-implication}
   If the conjectured lower bound
   \[
      \Tcost_T(m) = \Omega(m \lg m)
   \]
   holds for every binary transposition machine $T$,
   then the conjectured lower bound
   \[
      \Mcost_M(m) = \Omega(m \lg m)
   \]
   holds for every multiplication machine $M$.
\end{thm}

In fact, we can prove something slightly stronger than this:
provided that $\Tcost(m) = \omega(m \lg \lg m)$,
then roughly speaking $\Mcost(m) = \Omega(\Tcost(m))$,
i.e., multiplying integers is at least as hard as transposing
matrices of the same bit size.
For a precise statement, see Corollary \ref{cor:main}.

If even the lower bound $\Tcost(m) = \omega(m \lg \lg m)$ is out of reach,
then we can still say something,
but our results are somewhat weaker.
For example, if $\Tcost(m) = \omega(m \lg \lg \lg m)$,
then we would like to be able to prove that
$\Mcost(m) = \omega(m \lg \lg \lg m)$,
but we cannot quite manage this.
Instead we have the following weaker statement.
\begin{thm}
   \label{thm:main-ell-implication}
   Fix $\ell \geq 2$.
   If no binary transposition machine achieves
   $\Tcost(m) = O(m \lg^{\circ \ell} m)$,
   then no multiplication machine achieves
   $\Mcost(m) = O(m \lg^{\circ \ell} m)$.
\end{thm}

Finally, if we cannot even rule out the possibility that
$\Tcost(m) = O(m \lg^{\circ \ell} m)$ is feasible for every fixed $\ell$,
then we make our last stand with the following result.
\begin{thm}
   \label{thm:main-recursive-implication}
   Let $f(m)$ be a non-decreasing function.
   If no binary transposition machine achieves
   $\Tcost(m) = O(m f(m) \lg^* m)$,
   then no multiplication machine achieves $\Mcost(m) = O(m f(m))$.
\end{thm}
Theorem \ref{thm:main-recursive-implication} implies for example
that if no transposition machine achieves $\Tcost(m) = O(m \lg^* m)$,
then multiplication cannot be performed in linear time.
We would prefer to be able to prove the stronger statement that
if transposition cannot be carried out in linear time,
then neither can multiplication,
but our methods do not appear to be strong enough to prove this.
(See however Theorem \ref{thm:log-implication},
which gives a result of similar strength for matrices with
larger coefficients.)

Theorems \ref{thm:main-implication}--\ref{thm:main-recursive-implication}
suggest that instead of trying to prove lower bounds
for multiplication directly,
it might be preferable to attack the transposition problem.
After all, transposition merely moves data around;
it does not seem to involve any actual \emph{computation}.
Unfortunately, as mentioned earlier,
we still have no idea how to prove non-trivial lower bounds
for transposition in the Turing model.
We hope that the results of this paper will inspire
further research on this question.

\subsection{Overview of methods}
\label{sec:overview}

The core of our reduction runs as follows.
Suppose that we wish to transpose an
$n_1 \times n_2$ matrix with $b$-bit entries.
We interpret the matrix as a vector in $\CC^{n_1 n_2}$,
regarding its entries as $b$-bit fixed-point approximations
to complex numbers.
The idea is to compute the DFT (discrete Fourier transform)
of this vector using one algorithm,
and then compute the inverse DFT of the result using a different algorithm,
in such a way that the original data comes back in transposed order.

For the DFT in the forward direction,
we use Bluestein's method \cite{Blu-dft} to reduce the DFT
to a convolution problem,
and then Kronecker substitution \cite[Cor.~8.27]{vzGG-compalg3}
to reduce the convolution to a large integer multiplication problem.
The details of the Bluestein--Kronecker combination
are worked out in Section \ref{sec:bluestein-kronecker}.

For the inverse DFT,
we first use the Cooley--Tukey method \cite{CT-fft}
to decompose the transform into DFTs along the rows and columns of the matrix,
plus a collection of multiplications by ``twiddle factors''.
We then handle the row and column transforms
using the same Bluestein--Kronecker combination mentioned above.
This strategy is explained in Section \ref{sec:main}.
The reason that it succeeds in transposing the data
ultimately comes down to algebraic properties
of the Cooley--Tukey decomposition.
(This is not unrelated to the familiar fact that textbook
``decimation-in-time'' and ``decimation-in-frequency'' FFT algorithms
naturally produce their output in bit-reversed order.)

In all of the above steps,
it is crucial to carefully manage the ordering of data on the tape.
For instance, we rely heavily on the fact that Bluestein's trick
computes the DFT in the natural ordering;
it would not work to substitute Rader's algorithm \cite{Rad-prime},
which involves complicated permutations of the data.
Another example is that in the Cooley--Tukey decomposition,
we cannot afford to transpose the data so that the column transforms
operate on contiguous data;
instead we must show how to carry out the Kronecker substitution for
all columns in parallel, without reordering the data.

Throughout the computation, we must ensure that all numerical computations
are performed with sufficient accuracy to recover the correct (integer)
results at the end of the transposition.
This entails some straightforward but tedious numerical analysis.
To facilitate this,
in Section \ref{sec:fixed-point} we introduce a framework
for analysing numerical error in fixed-point computations,
along the lines of \cite[\S2]{HvdH-nlogn}.

The approach described so far has a serious limitation:
it only works when the coefficient size $b$ is sufficiently large
compared to the matrix dimensions $n_1$ and $n_2$.
This occurs for several reasons,
including coefficient growth in the Kronecker substitution step,
and the need to be able to represent complex $(n_1 n_2)$-th roots of unity
with sufficient accuracy.
As a partial workaround for this problem,
in Section~\ref{sec:decompose} we explain
how to decompose a given transposition problem
into a collection of transposition problems
of exponentially smaller dimension.
In Section \ref{sec:proofs} we show how to use this strategy
to handle transposition problems involving smaller coefficient sizes,
all the way down to single-bit coefficients,
and thereby finish the proofs of the results
stated in Section \ref{sec:summary}.

\begin{rem}
   We take this opportunity to correct a missing attribution
   in one of our earlier papers.
   The Bluestein--Kronecker combination was announced in \cite{HvdHL-mul},
   but in fact, as pointed out to us by Dan Bernstein
   (personal communication, 2022),
   it goes back at least to \cite[\S3]{Sch-numerical}.
\end{rem}

\subsection{Known lower bounds}
\label{sec:known-lower-bounds}

In the standard Turing model, virtually no non-trivial (i.e.~superlinear)
lower bounds are known for basic computational tasks such as
transposition, integer multiplication, FFTs, finding duplicate
elements in a list, sorting, and so on~\cite{Reg-superlinear}.
Interesting lower bounds have been proved for the Turing model
under additional restrictions and for various other complexity models.

Perhaps the most impressive lower bound result for multiplication
is the $\Omega(n \log n)$ proved for ``online'' (or relaxed)
integer multiplication~\cite{PFM-overlap,CA-minimum},
where we restrict the bits of the input and output
to be given and computed ``one by one''.
An almost matching upper bound $O(n \log n \, e^{O(\sqrt{\log \log n})})$
is also known for this problem~\cite{vdH:fastrelax}.
In the Boolean circuit model, a $\Omega(n \log n)$ lower bound
was recently obtained as well, conditional on an open conjecture
on network coding~\cite{afshani2019}.
For the related problem of FFT computations,
Morgenstern proved an $\Omega (n \log n)$ lower bound
for evaluating a complex DFT of length $n$
in a suitably restricted algebraic complexity model~\cite{Morgenstern73}.

Concerning the binary $n \times n$ matrix transposition problem,
a natural restriction is to consider Turing machines that
can only move data, but not perform any actual computations.
In such a model, Sto{\ss} proved the expected $\Omega (n^2 \log n)$
lower bound~\cite{Stoss-rangierkomplexitat}. Non-trivial
lower bounds are also known for a few less natural complexity models,
such as a two-level fast-slow memory model~\cite{Floyd-permuting}
or Turing machines with a single
tape~\cite{Kirchherr-transposition,DM93-transposition}.
Finally, conditional on the same network coding conjecture as above,
an $\Omega (n^2 \log n)$ lower bound has been proved
for Turing machines with two tapes~\cite{AHJKL06}.

\section{Fixed-point arithmetic}
\label{sec:fixed-point}

Let $\disc \coloneqq \{u \in \CC : |u| \leq 1 \}$ be the complex unit disc.
Let $p \geq 1$ be a precision parameter, and define
\[
   \appdisc \coloneqq (2^{-p} \, \ZZ[\ii]) \cap \disc
   = \{2^{-p} (x + \ii y) : x, y \in \ZZ \text{ and } x^2 + y^2 \leq 2^{2p}\}.
\]
We will regard elements of $\appdisc$ as finite-precision approximations
for elements of $\disc$.
An element of $\appdisc$ will be represented on the tape by
the pair of integers $(x,y)$,
and occupies $O(p)$ space.

If $z \in \disc$, we systematically write $\app z \in \appdisc$
for a fixed-point approximation for $z$
that has been computed by some algorithm.
We write
\[
   \varepsilon(\app z) \coloneqq 2^p \, |\app z - z|
\]
for the associated error, measured as a multiple of $2^{-p}$
(the ``unit in the last place'').

We define a round-towards-zero function $\rho \colon \CC \to \CC$ as follows.
First, define $\rho_0 \colon \RR \to \ZZ$ by
$\rho(x) \coloneqq \lfloor x \rfloor$ for $x \geq 0$ and
$\rho(x) \coloneqq \lceil x \rceil$ for $x < 0$.
Then define $\rho_0 \colon \CC \to \ZZ[\ii]$ by setting
$\rho_0(x + \ii y) \coloneqq \rho_0(x) + \ii \rho_0(y)$.
Finally, set $\rho(z) \coloneqq 2^{-p} \rho_0(2^p z)$ for $z \in \CC$.
In other words, $\rho(z)$ rounds the real and imaginary parts of $z$
to the nearest multiple of $2^{-p}$ in the direction of the origin.
Clearly $|\rho(z)| \leq |z|$ for all $z \in \CC$,
so $\rho$ maps $\disc$ to $\appdisc$, and also
\begin{equation}
   \label{eq:rho-bound}
   |\rho(z) - z| \leq \sqrt 2 \cdot 2^{-p}, \qquad z \in \CC.
\end{equation}

The following result is almost identical to \cite[Cor.~2.10]{HvdH-nlogn}.
\begin{lem}[Fixed point multiplication]
   \label{lem:mul}
   Given any multiplication machine $M$,
   there exists a Turing machine with the following properties.
   Its input consists of the precision parameter $p$
   and approximations $\app u, \app v \in \appdisc$ for $u, v \in \disc$.
   Its output is an approximation $\app w \in \appdisc$ for
   $w \coloneq uv \in \disc$ such that
   $\varepsilon(\app w) < \varepsilon(\app u) + \varepsilon(\app v) + 2$.
   Its running time is $O(\Mcost_M(p))$.
\end{lem}
\begin{proof}
   We first compute $\app u \app v \in \disc$ by multiplying
   out the real and imaginary parts of $\app u$ and $\app v$
   and summing appropriately.
   We then define $\app w \coloneqq \rho(\app u \app v) \in \appdisc$.
   Applying~$\rho$ amounts to a simple rounding operation,
   so the complexity of computing $\app w$ is $4 \Mcost_M(p) + O(p) = O(\Mcost_M(p))$.
   As for the error bound, observe that
   \[
      \error(\app w) = 2^p \, | \app w - w |
         \leq 2^p \, |\rho(\app u \app v) - \app u \app v|
            + 2^p \, | \app u \app v - u v|.
   \]
   We have $2^p \, |\rho(\app u \app v) - \app u \app v| \leq \sqrt 2$ by
   \eqref{eq:rho-bound}.
   For the second term,
   \begin{align*}
      2^p \, |\app u \app v - uv|
         \leq 2^p \, |\app u \app v - u \app v|
            + 2^p \, |u \app v - u v|
         & = 2^p \, |\app u - u| \cdot |\app v|
            + |u| \cdot 2^p \, |\app v - v| \\
         & \leq \error(u) \cdot 1 + 1 \cdot \error(v).
   \end{align*}
   Thus $\error(\app w) \leq \sqrt 2 + \error(u) + \error(v)
      < \error(u) + \error(v) + 2$.
\end{proof}

We next consider the computation of roots of unity.
For $n \geq 1$, define
\[
   \zeta_n \coloneqq e^{2 \pi \ii / n} \in \disc.
\]
\begin{lem}[Roots of unity]
   \label{lem:zeta}
   For any constant $C > 0$, there exists a Turing machine
   with the following properties.
   Its input consists of the precision parameter~$p$
   and a positive integer $n < 2^{C p}$.
   Its output is an approximation $\app \zeta_n \in \appdisc$ for $\zeta_n$
   with $\error(\app\zeta_n) < 2$.
   Its running time is $O(p \lg^2 p)$.
\end{lem}
(In the above lemma, and in similar situations throughout the paper,
the big-$O$ constant depends on $C$.
The symbol $C$ will be reused many times;
it does not refer to the same constant each time.)
\begin{proof}
   Let $q \coloneqq p + C'$ for a suitable constant $C' > 0$.
   Using standard methods
   (see for example \cite[Chs.~6--7]{BB-pi-agm} or \cite[Ch.~4]{BZ-mca})
   together with the multiplication machine $M^*$
   (whose complexity is given by \eqref{eq:Mstar-bound}),
   we may compute a $q$-bit approximation to $\pi$ in time $O(q \lg^2 q)$,
   a $q$-bit approximation to $1/n$ in time $O(q \lg q)$,
   and then $q$-bit approximations to $\cos(2 \pi / n)$ and $\sin(2 \pi / n)$
   in time $O(q \lg^2 q)$.
   Rounding towards the origin at precision $p$,
   we obtain the desired error bound.
   For further details, see \cite[\S2.8]{HvdH-nlogn}.
\end{proof}
\begin{rem}
   \label{rem:precomputations}
   The reader may wonder why we used $M^*$ in Lemma \ref{lem:zeta},
   i.e., why could we not write the complexity bound as
   $O(\Mcost_M(p) \log p)$ for a given multiplication machine $M$,
   along similar lines to Lemma \ref{lem:mul}?
   The reason is that the standard textbook results
   rely on properties of $\Mcost_M(m)$
   that may not be satisfied for an arbitrary
   black-box multiplication machine $M$.
   Consider for example the reciprocal step in the above proof,
   i.e., computing $1/n$.
   In the typical Newton iteration scheme for the reciprocal,
   the precision doubles at each step, so the complexity is really
   $O(\Mcost(q) + \Mcost(q/2) + \Mcost(q/4) + \cdots)$.
   This can only be simplified to $O(\Mcost(q))$ under additional assumptions,
   such as $\Mcost(q)/q$ being non-decreasing.
   This holds for $M^*$ but we cannot guarantee that it holds for $M$.
\end{rem}

\section{DFTs and the Bluestein--Kronecker trick}
\label{sec:bluestein-kronecker}

\subsection{Conventions for vectors}

If $U \in \CC^n$, we write $U_0, \ldots, U_{n-1}$ for the components of $U$.
We write $\disc^n$ to denote the subset of $\CC^n$ whose entries lie in~$\disc$,
and similarly for $\appdisc^n$.
Elements of $\appdisc^n$ are always stored on the tape
in the standard ordering $U_0, \ldots, U_{n-1}$.
If $U = (U_0, \ldots, U_{n-1}) \in \disc^n$ is a vector,
and $\app U = (\app U_0, \ldots, \app U_{n-1}) \in \appdisc^n$
is an approximation, we write $\error(\app U)$ for $\max_i \error(\app U_i)$.

\subsection{Discrete Fourier transforms}

Let us recall the definition of the DFT of length $n \geq 1$ over $\CC$.
Given a vector $X \in \CC^n$, its DFT is the vector $Y \in \CC^n$ defined by
\begin{equation}
   \label{eq:dft}
   Y_t \coloneqq \frac1n \sum_{s=0}^{n-1} \zeta_n^{-st} X_s,
      \qquad 0 \leq t < n,
\end{equation}
where we recall that $\zeta_n \coloneqq e^{2 \pi \ii / n} \in \disc$.
Note the slightly nonstandard scaling factor $1/n$,
which we have included to ensure that the DFT maps $\disc^n$ into $\disc^n$.

It is straightforward to check that the inverse transform is given by
\begin{equation}
   \label{eq:inverse-dft}
   X_s = \sum_{t=0}^{n-1} \zeta_n^{st} Y_t,
      \qquad 0 \leq s < n.
\end{equation}
This formula is the same as the formula for the forward DFT,
except for a sign change and the absence of the scaling factor.

\subsection{Bluestein's trick}
\label{sec:bluestein}

We now recall Bluestein's reduction from DFTs to convolutions \cite{Blu-dft}.
Let $n \geq 1$ and put $\zeta \coloneqq \zeta_n$.
Let $X \in \CC^n$, and let $Y \in \CC^n$ be the DFT of $X$
according to \eqref{eq:dft}.
For any $k \in \ZZ$, define
\begin{equation}
   \label{eq:omega}
   \omega_k \coloneqq \zeta^{\binom{k}{2}} = \zeta^{k(k-1)/2}.
\end{equation}
Then the identity $st = \binom{t}{2} + \binom{-s}{2} - \binom{t-s}{2}$
implies that
\[
   Y_t = \frac1n \sum_{s=0}^{n-1} \zeta^{-st} X_s
      = \frac1n \, \bar\omega_t \sum_{s=0}^{n-1}
      (\bar\omega_{-s} X_s) \cdot \omega_{t-s},
      \qquad 0 \leq t < n.
\]
(The bar $\bar \cdot$ denotes complex conjugation.)
The latter sum may be recognised as a portion of an acyclic convolution
of a sequence of length $n$ with a sequence of length $2n-1$.
More explicitly,
if we define vectors $X', Y' \in \CC^n$ and $B \in \CC^{2n-1}$ by
\begin{alignat}{2}
   \label{eq:Xprime}
   X'_s & \coloneqq \bar\omega_{-s} X_s, & \qquad & 0 \leq s < n, \\
   \label{eq:Yprime}
   Y'_t & \coloneqq \omega_t Y_t,        &        & 0 \leq t < n, \\
   \label{eq:B}
   B_r  & \coloneqq \omega_{r-n+1},      &        & 0 \leq r < 2n-1,
\end{alignat}
then the relation becomes
\begin{equation}
   \label{eq:bluestein}
   Y'_t = \frac1n \sum_{s=0}^{n-1} X'_s B_{(t+n-1)-s}, \qquad 0 \leq t < n.
\end{equation}
Therefore,
the desired DFT may be computed by first using \eqref{eq:Xprime}
to compute the vector $X'$ from $X$,
then computing the convolution in \eqref{eq:bluestein} to obtain $Y'$,
and finally deducing $Y$ from $Y'$ via \eqref{eq:Yprime}.

\subsection{Bluestein plus Kronecker substitution}

The following result shows how to use integer multiplication to evaluate
a sum of the type \eqref{eq:bluestein}.
The idea is to perform a Kronecker substitution
\cite[Cor.~8.27]{vzGG-compalg3},
i.e., pack the coefficients of the sequences into large integers,
multiply the resulting integers,
and then extract the convolution from the resulting integer product.
Recall that $p$ denotes the fixed-point precision
as in Section \ref{sec:fixed-point}.

\begin{prop}[One-dimensional convolution]
   \label{prop:convolution}
   Let $C, C' > 0$ be constants,
   and let $M, M'$ be multiplication machines.
   Then there exists a Turing machine with the following properties.
   Its input consists of the precision parameter $p$,
   positive integers $n$ and $m$ such that
   \[
      \lg n < Cp, \qquad \lg p < C'n, \qquad np \leq m,
   \]
   and approximations $\app F \in \appdisc^n$ and $\app G \in \appdisc^{2n-1}$
   for vectors $F \in \disc^n$ and $G \in \disc^{2n-1}$.
   Define $H \in \disc^n$ by the formula
   \[
      H_t \coloneqq \frac1n \sum_{s=0}^{n-1} F_s G_{(t+n-1)-s},
         \qquad 0 \leq t < n.
   \]
   The output of the machine is an approximation $\app H \in \appdisc^n$
   such that $\error(\app H) < \error(\app F) + \error(\app G) + 2$.
   Its running time is
   \[
      O(\Mcost_M(m) + n \Mcost_{M'}(p)).
   \]
\end{prop}
Before giving the proof,
we mention a simple fact that will be used frequently:
\begin{lem}
   \label{lem:chunks}
   Let $C > 1$ be a constant and let $M$ be a multiplication machine.
   Then there exists a Turing machine that takes as input positive
   integers $m$ and $m' \leq Cm$,
   and two $m'$-bit integers $x$ and $y$,
   and returns the product $xy$ in time $O(\Mcost_M(m))$.
\end{lem}
\begin{proof}
   Simply cut up each multiplicand of size $m'$ into
   $\lceil m'/m \rceil = O(1)$ chunks of size $m$,
   multiply out the chunks, and sum appropriately.
\end{proof}
Loosely speaking, Lemma \ref{lem:chunks} says that
if $m' = O(m)$ then $\Mcost(m') = O(\Mcost(m))$.
However, the latter statement is not quite correct,
because if the multiplication machine $M$ is given inputs of size $m'$,
it does not ``know'' that it should reduce to problems of size $m$.
The point of Lemma \ref{lem:chunks} is to embed $M$ in a machine
that explicitly requires both $m$ and $m'$ as input.

\begin{proof}[Proof of Proposition \ref{prop:convolution}]
   Recall that $2^p \app F_s, 2^p \app G_j \in \ZZ[\ii]$.
   Consider the polynomials
   \[
      f(x) \coloneqq \sum_{s=0}^{n-1} (2^p \app F_s) x^s
      \in \ZZ[\ii][x], \qquad
      g(x) \coloneqq \sum_{j=0}^{2n-2} (2^p \app G_j) x^j
      \in \ZZ[\ii][x],
   \]
   and let $w \coloneqq fg \in \ZZ[\ii][x]$ be their product,
   say $w(x) = \sum_{k=0}^{3n-3} W_k x^k$.
   Since $|2^p \app F_s|, |2^p \app G_j| \leq 2^p$,
   we have the bound $|W_k| \leq n (2^p)^2$ for all $k$.
   Let
   \[
      \beta \coloneqq 2p + \lg n + 2,
   \]
   so that $|W_k| \leq 2^{\beta-2}$.
   Note that the hypothesis $\lg n < Cp$ implies that $\beta = O(p)$.
   
   We first compute $f(2^\beta), g(2^\beta) \in \ZZ[\ii]$,
   by concatenating the input coefficients with appropriate zero-padding
   (or one-padding in the case of negative coefficients),
   handling the real and imaginary parts separately.
   This requires time $O(n\beta) = O(np)$,
   including the cost of various carry/borrow handling,
   the details of which are omitted.
   
   We next compute the product $w(2^\beta) = f(2^\beta) g(2^\beta)$
   by using $M$ to multiply out the real and imaginary parts
   of $f(2^\beta)$ and $g(2^\beta)$, working in chunks of size $m$,
   and then adding/subtracting appropriately.
   Each multiplicand has $n\beta = O(np) = O(m)$ bits,
   so by Lemma \ref{lem:chunks} the cost is $O(\Mcost_M(m))$.
   
   Since the real and imaginary parts of $W_k$ are bounded
   in absolute value by $2^{\beta-2}$,
   we may extract all the $W_k$ unambiguously from $w(2^\beta)$
   in time $O(np)$,
   i.e., we first examine $w(2^\beta) \bmod 2^\beta$ to recover $W_0$,
   then we remove $W_0$ from the sum and continue onto $W_1$, and so on.
   
   At this stage we have recovered the coefficients of interest,
   namely
   \[
      W_{t+n-1} = \sum_{s=0}^{n-1}
         (2^p \app F_s) (2^p \app G_{(t+n-1)-s}) \in \ZZ[\ii],
            \qquad 0 \leq t < n,
   \]
   and we define the output approximations for $H_t$ to be
   \[
      \app H_t \coloneqq \rho(W_{t+n-1}/2^{2p} n) \in \appdisc,
         \qquad 0 \leq t < n.
   \]
   This is well defined, as we showed earlier that
   $|W_{t+n-1}/2^{2p} n| \leq 1$.
   Computing each $\app H_t$ from $W_{t+n-1}$ requires a division by $n$
   followed by appropriate rounding.
   The cost of the rounding is $O(p)$.
   For the division by $n$,
   we first use the multiplication machine $M^*$
   to precompute a $p$-bit approximation for $1/n$ in time $O(p \lg p)$
   (see Remark \ref{rem:precomputations}).
   This bound simplifies to $O(np)$ thanks to the hypothesis
   $\lg p < C'n$.
   The division by $n$ then reduces to a $p$-bit multiplication
   (plus additional linear-time work);
   using $M'$ to compute this product, the cost is $O(\Mcost_{M'}(p))$
   for each $t$, and hence $O(n \Mcost_{M'}(p))$ altogether.
   
   Finally, observe that
   \begin{multline*}
      \left| \frac{W_{t+n-1}}{2^p n} - 2^p H_t \right|
         = \frac{2^p}{n} \left| \sum_{s=0}^{n-1} \big(
            \app F_s \app G_{(t+n-1)-s} - F_s G_{(t+n-1)-s} \big) \right| \\
         \leq \frac{2^p}{n} \sum_{s=0}^{n-1}
            \left( | \app F_s - F_s | | \app G_{(t+n-1)-s} |
               + | F_s | | \app G_{(t+n-1)-s} - G_{(t+n-1)-s} | \right) \\
         \leq \frac1n \sum_{s=0}^{n-1} \left(
            \error(\app F_s) \cdot 1 + 1 \cdot \error(\app G_{(t+n-1)-s})
            \right)
            \leq \error(\app F) + \error(\app G),
   \end{multline*}
   and hence, by \eqref{eq:rho-bound},
   \begin{align*}
      \error(\app H_t) & = 2^p |\app H_t - H_t| \\
      & \leq 2^p \left| \rho\left(\frac{W_{t+n-1}}{2^{2p} n}\right)
         - \frac{W_{t+n-1}}{2^{2p} n} \right|
         + 2^p \left| \frac{W_{t+n-1}}{2^{2p} n} - H_t \right| \\
      & \leq \sqrt 2 + \error(\app F) + \error(\app G)
   \end{align*}
   for all $t$.
\end{proof}

\begin{prop}[One-dimensional DFT]
   \label{prop:dft}
   Let $C, C' > 0$ be constants,
   and let $M, M'$ be multiplication machines.
   Then there exists a Turing machine with the following properties.
   Its input consists of the precision parameter $p$,
   positive integers $n$ and $m$ such that
   \[
      \lg n < Cp, \qquad \lg^2 p < C' n, \qquad np \leq m,
   \]
   and an approximation $\app X \in \appdisc^n$ for a vector $X \in \disc^n$.
   Let $Y \in \disc^n$ be the DFT of $X$ according to \eqref{eq:dft}.
   The output of the machine is an approximation $\app Y \in \appdisc^n$
   such that $\error(\app Y) < \error(\app X) + 12 n^2$.
   Its running time is
   \[
      O(\Mcost_M(m) + n \Mcost_{M'}(p)).
   \]
\end{prop}

\begin{rem}
   Let us comment briefly on the significance of the hypotheses
   in Proposition \ref{prop:dft}.
   The condition $\lg n < Cp$ was already alluded to
   in Section \ref{sec:overview}, and seems to be unavoidable;
   the reduction simply does not work if $p$ is too small relative to $n$.
   By contrast, the hypothesis $\lg^2 p < C' n$ is merely a technical annoyance.
   It is included to control the cost of computing roots of unity,
   but we will see in Section~\ref{sec:proofs}
   that it does not cause any serious problems,
   because it is always possible to cut up large coefficients
   into smaller chunks without any detrimental effect.
\end{rem}

\begin{proof}[Proof of Proposition \ref{prop:dft}]
   The case $n = 1$ is trivial, so we assume that $n \geq 2$.
   We will use the strategy described in Section \ref{sec:bluestein}.
   
   \step{Step 1: compute $\zeta$.}
   Invoking Lemma \ref{lem:zeta},
   we compute an approximation $\app\zeta \in \appdisc$ for
   $\zeta \coloneqq \zeta_n = e^{2 \pi \ii / n}$
   with $\error(\app\zeta) < 2$
   in time $O(p \lg^2 p) = O(\Mcost_{M'}(p) \lg^2 p)$.
   Using the hypothesis $\lg^2 p < C'n$,
   this simplifies to $O(n \Mcost_{M'}(p))$.
   
   \step{Step 2: compute powers of $\zeta$.}
   Define $\gamma_j \coloneqq \zeta^j \in \disc$.
   We compute approximations $\app\gamma_j \in \appdisc$
   for $j = 0, 1, \ldots, n-1$ by taking $\app\gamma_0 \coloneqq 1$,
   $\app\gamma_1 \coloneqq \app\zeta$,
   and then repeatedly using Lemma \ref{lem:mul}
   together with the identity $\gamma_j = \gamma_{j-1} \cdot \zeta$.
   This requires time $O(n \Mcost_{M'}(p))$ (using the machine $M'$),
   and by induction we obtain $\error(\app\gamma_j) < 4j-2$
   for all $j = 1, \ldots, n-1$.
   We conclude that $\error(\app\gamma_j) < 4(n-1)-2 = 4n-6$
   for all $j = 0, \ldots, n-1$.

   \step{Step 3: compute $\omega_k$.}
   Similarly, we compute approximations $\app\omega_k \in \appdisc$
   for $\omega_k \in \disc$ (defined in \eqref{eq:omega}),
   for $k = 0, \ldots, n$,
   by taking $\app\omega_0 \coloneqq 1$
   and then repeatedly applying the identity
   $\omega_k = \omega_{k-1} \cdot \gamma_{k-1}$.
   Again this requires time $O(n \Mcost_{M'}(p))$,
   and we obtain by induction $\error(\app\omega_k) < 4k(n-1)$
   for $k = 1, \ldots, n$.
   Therefore $\error(\app\omega_k) < 4n(n-1)$
   for all $k = 0, \ldots, n$.

   \step{Step 4: compute $B_r$.}
   We compute an approximation $\app B \in \appdisc^{2n-1}$
   for the vector $B \in \disc^{2n-1}$ defined in \eqref{eq:B}.
   The entries $B_r = \omega_{r-n+1}$ are obtained directly
   from the output of Step 3,
   using the fact that $\omega_{-s} = \omega_{s+1}$
   (since $\binom{-s}{2} = \binom{s+1}{2}$).
   This requires time $O(np) = O(n \Mcost_{M'}(p))$,
   and we obtain $\error(\app B) < 4n(n-1)$.

   \step{Step 5: compute $X'_s$.}
   We use Lemma \ref{lem:mul}
   to compute an approximation $\app X' \in \appdisc^n$
   for $X' \in \disc^n$ defined by $X'_s = \bar\omega_{-s} X_s$
   (see \eqref{eq:Xprime}).
   This requires time $O(n \Mcost_{M'}(p))$,
   and we obtain $\error(\app X') < \error(\app X) + 4n(n-1) + 2$.
   
   \step{Step 6: perform convolution.}
   We apply Proposition \ref{prop:convolution}
   with $F \coloneqq X'$, $G \coloneqq B$, $H \coloneqq Y'$,
   to obtain an approximation $\app Y' \in \appdisc^n$
   for $Y' \in \disc^n$,
   where the entries $Y'_t$ are given by \eqref{eq:bluestein}.
   This requires time $O(\Mcost_M(m) + n \Mcost_{M'}(p))$,
   and we obtain $\error(\app Y') < \error(\app X') + \error(\app B) + 2
   < \error(\app X) + 8n(n-1) + 4$.
   
   \step{Step 7: compute $Y_t$.}
   We finally approximate $Y \in \disc^n$ via $Y_t = \bar\omega_t Y'_t$
   (see \eqref{eq:Yprime}).
   Using Lemma \ref{lem:mul} this requires time $O(n \Mcost_{M'}(p))$,
   and we obtain $\error(\app Y) < \error(\app Y') + 4n(n-1) + 2
   < \error(\app X) + 12n(n-1) + 6 < \error(\app X) + 12 n^2$.
\end{proof}

\begin{rem}
   The same algorithm works if we want to evaluate the transform
   $\frac1n \sum_{s=0}^{n-1} \zeta_n^{st} X_s$,
   i.e., replacing $\zeta_n^{-1}$ by $\zeta_n$ in \eqref{eq:dft}.
   We use this fact several times below without further comment.
\end{rem}

\begin{rem}
   In the above proof,
   instead of computing the $\omega_k$ via multiplication,
   one could extract them from the list of $\zeta^j$
   by means of a sorting algorithm.
   (As a side-effect, the numerical error would be reduced from
   $O(n^2)$ to $O(n)$.)
   However, the best complexity bound we know for this approach is
   $O(p n \lg n)$ via a merge sort \cite{Knu-TAOCP3},
   introducing an unwanted $\lg n$ factor into the final result.
\end{rem}

In the main reduction in Section \ref{sec:main},
we will need to apply the Bluestein--Kronecker method
to compute transforms along various one-dimensional ``slices'' of
higher-dimensional arrays.
In the remainder of this section we explain how to generalise
Propositions \ref{prop:convolution} and~\ref{prop:dft}
to handle this case.

\subsection{Conventions for arrays}

We will write $\CC^{n_1,\ldots,n_d}$ for the tensor product
$\CC^{n_1} \otimes \cdots \otimes \CC^{n_d}$,
i.e., the set of $d$-dimensional arrays of size
$n_1 \times \cdots \times n_d$ with entries in~$\CC$.
We write similarly $\disc^{n_1,\ldots,n_d}$ and $\appdisc^{n_1,\ldots,n_d}$
for the subsets of $\CC^{n_1,\ldots,n_d}$
with coefficients lying in $\disc$ or $\appdisc$.
For an array $U \in \CC^{n_1,\ldots,n_d}$,
we write $U_{s_1,\ldots,s_d}$ for the entry at position $(s_1, \ldots, s_d)$,
where $0 \leq s_i < n_i$ for $i = 1, \ldots, d$.
The entries of an array $U \in \appdisc^{n_1,\ldots,n_d}$ are
always stored on the tape in row-major order,
i.e., as a list of length $n_1 \cdots n_d$,
with $U_{s_1,\ldots,s_d}$ stored at index
$s_1 (n_2 \cdots n_d) + s_2 (n_3 \cdots n_d) + \cdots + s_{d-1} n_d + s_d$.

\subsection{Transforms along slices}

Let $l_1, l_2, n \geq 1$.
Suppose that we are given an $l_1 \times n \times l_2$ array
$X \in \disc^{l_1,n,l_2}$,
and we wish to compute the transform $Y \in \disc^{l_1,n,l_2}$ given by
\begin{equation}
   \label{eq:dft-slices}
   Y_{i_1,t,i_2} \coloneqq \frac1n
   \sum_{s=0}^{n-1} \zeta_n^{-st} X_{i_1,s,i_2},
   \qquad 0 \leq i_1 < l_1, \;\; 0 \leq t < n, \;\; 0 \leq i_2 < l_2.
\end{equation}
In other words, for each $(i_1,i_2)$,
we want to compute a transform of length $n$ along the
slice indexed by $(i_1,*,i_2)$.
We cannot simply apply Proposition \ref{prop:dft} to each slice separately,
because that would require first transposing the data,
which is precisely what we are trying to avoid.
Instead, we will show how to encode all of the DFTs
into a single large integer multiplication problem.

We begin by generalising Proposition \ref{prop:convolution}
to handle this situation.
\begin{prop}[Convolutions along slices]
   \label{prop:convolution-slices}
   Let $C, C' > 0$ be constants, and let $M, M'$ be multiplication machines.
   Then there exists a Turing machine with the following properties.
   Its input consists of the precision parameter $p$,
   positive integers $l_1, l_2, n, m$ such that
   \[
      \lg n < Cp, \qquad \lg p < C' l_1 l_2 n, \qquad l_1 l_2 np \leq m,
   \]
   and approximations $\app F \in \appdisc^{l_1,n,l_2}$
   and $\app G \in \appdisc^{2n-1}$
   for an array $F \in \disc^{l_1,n,l_2}$ and a vector $G \in \disc^{2n-1}$.
   Define $H \in \disc^{l_1,n,l_2}$ by the formula
   \[
      H_{i_1,t,i_2} \coloneqq \frac1n
      \sum_{s=0}^{n-1} F_{i_1,s,i_2} G_{(t+n-1)-s},
      \qquad 0 \leq i_1 < l_1, \;\; 0 \leq t < n, \;\; 0 \leq i_2 < l_2.
   \]
   The output of the machine is an approximation
   $\app H \in \appdisc^{l_1,n,l_2}$ such that
   $\error(\app H) < \error(\app F) + \error(\app G) + 2$.
   Its running time is
   \[
      O(\Mcost_M(m) + l_1 l_2 n \Mcost_{M'}(p)).
   \]
\end{prop}
\begin{proof}
   The proof is similar to the proof of Proposition \ref{prop:convolution},
   working instead with the polynomials
   \begin{align*}
      f(x) & \coloneqq
         \sum_{i_1=0}^{l_1-1} \sum_{s=0}^{n-1} \sum_{i_2=0}^{l_2-1}
         (2^p \app F_{i_1,s,i_2}) x^{3nl_2 i_1 + l_2 s + i_2}
      \in \ZZ[\ii][x], \\
      g(x) & \coloneqq \sum_{j=0}^{2n-2} (2^p \app G_j) x^{l_2 j}
      \in \ZZ[\ii][x].
   \end{align*}
   We claim that the coefficients of the product $f(x) g(x)$ include
   the desired output coefficients, in the correct order.
   Indeed, we may rewrite $f(x) g(x)$ as
   \[
      f(x) g(x) = \sum_{i_1=0}^{l_1-1} \sum_{i_2=0}^{l_2-1}
         \left( \sum_{s=0}^{n-1} \sum_{j=0}^{2n-2}
            (2^p \app F_{i_1,s,i_2}) (2^p \app G_j) x^{l_2(s+j)} \right)
            x^{3nl_2 i_1 + i_2}.
   \]
   The expression within the large parentheses is a polynomial
   of degree at most $3n-3$ in $x^{l_2}$.
   For each $(i_1,i_2)$, it corresponds to the convolution of
   $F_{i_1,*,i_2}$ with~$G_*$.
   The $x^{3nl_2 i_1 + i_2}$ term at the end
   ensures that these polynomials do not overlap in $f(x) g(x)$.
   (Actually, the algorithm would still work if $3$ were replaced by $2$,
   since we are only interested in the ``middle third'' of each block.)
   
   Briefly, the algorithm runs as follows.
   We first evaluate at $x = 2^\beta$, taking
   $\beta \coloneqq 2p + \lg n + 2 = O(p)$.
   Evaluating $f(2^\beta)$ and $g(2^\beta)$ requires time $O(l_1 l_2 n p)$;
   the key point is that during evaluation of $f(2^\beta)$,
   the coefficients in the input array $\app F$ are already in the correct
   order for the desired concatenation.
   Computing the product $f(2^\beta) g(2^\beta)$ reduces to
   multiplying integers of bit size $O(l_1 l_2 n p) = O(m)$.
   Finally, during the decoding phase,
   the desired output approximations $\app H_{i_1,t,i_2}$
   are already in the correct order.
   We omit the details of the rounding and error analysis,
   which are essentially identical to Proposition \ref{prop:convolution}.
\end{proof}

\begin{prop}[DFTs along slices]
   \label{prop:dft-slices}
   Let $C, C' > 0$ be constants,
   and let $M, M'$ be multiplication machines.
   Then there exists a Turing machine with the following properties.
   Its input consists of the precision parameter $p$,
   positive integers $l_1, l_2, n, m$ such that
   \[
      \lg n < Cp, \qquad \lg^2 p < C' l_1 l_2 n, \qquad l_1 l_2 n p \leq m,
   \]
   and an approximation $\app X \in \appdisc^{l_1,n,l_2}$ for an array
   $X \in \disc^{l_1,n,l_2}$.
   Let $Y \in \disc^{l_1,n,l_2}$ be the DFT of $X$ along slices
   given by \eqref{eq:dft-slices}.
   The output of the machine is an approximation
   $\app Y \in \appdisc^{l_1,n,l_2}$
   such that $\error(\app Y) < \error(\app X) + 12 n^2$.
   Its running time is
   \[
      O(\Mcost_M(m) + l_1 l_2 n \Mcost_{M'}(p)).
   \]
\end{prop}
\begin{proof}
   The algorithm is essentially the same as in the proof of
   Proposition \ref{prop:dft}.
   We first execute steps 1--4 (computing various roots of unity).
   Next, we run step 5 (computing $X'_s$ from $X_s$)
   on each slice independently;
   this can be achieved by a single pass over the array.
   For step 6, we can perform the required convolutions on all slices
   simultaneously by replacing the use of Proposition~\ref{prop:convolution}
   with Proposition~\ref{prop:convolution-slices}.
   Finally, step 7 (computing $Y_t$ from $Y'_t$) can again
   be run independently on all slices in a single pass.
   The error analysis is identical to the proof of Proposition~\ref{prop:dft}.
\end{proof}

\begin{rem}
   \label{rem:combine}
   Proposition \ref{prop:dft-slices} may also be used to handle transforms
   along slices of higher-dimensional arrays.
   For example, if we have a 4-dimensional array of size
   $q_1 \times q_2 \times n \times q_3$,
   and we wish to transform with respect to the third coordinate,
   then we may reinterpret the data as a 3-dimensional array of size
   $(q_1 q_2) \times n \times q_3$,
   and invoke Proposition \ref{prop:dft-slices} with $l_1 \coloneqq q_1 q_2$
   and $l_2 \coloneqq q_3$.
   The complexity in this case would be
   $O(\Mcost_M(m) + q_1 q_2 q_3 n \Mcost_{M'}(p))$.
\end{rem}

\section{The main reduction}
\label{sec:main}

In this section we present the main reduction from
matrix transposition to integer multiplication,
in the case that the matrix entries are neither too small nor too large
relative to the matrix dimensions.
For reasons that will become clear in Section~\ref{sec:decompose},
we actually consider the following more general transposition problem.
Let $l_1, l_2, n_1, n_2, b \geq 1$,
and consider a 4-dimensional array $A$ of size
\[
   l_1 \times n_1 \times n_2 \times l_2,
\]
with $b$-bit entries.
We wish to transpose the $n_1$ and $n_2$ components, i.e.,
we want to compute the array $A'$ of size
\[
   l_1 \times n_2 \times n_1 \times l_2
\]
whose entries are given by
\[
   A'_{i_1,j_2,j_1,i_2} = A_{i_1,j_1,j_2,i_2}
   \qquad \text{for all $i_1, i_2, j_1, j_2$}.
\]
Equivalently, one may think of the input as a list of $l_1$ arrays
of size $n_1 \times n_2$ with $l_2 b$-bit entries,
and the goal is to transpose each of these $n_1 \times n_2$ arrays.
The case $l_1 = l_2 = 1$ corresponds to an ordinary
$n_1 \times n_2$ matrix transposition with $b$-bit entries.

A machine $T$ performing this type of transposition will be called a
\emph{generalised transposition machine}.
It takes as input the parameters $l_1, l_2, n_1, n_2, b$,
a positive integer $m$ such that $l_1 l_2 n_1 n_2 b \leq m$,
and the array to transpose.
We denote its worst case running time by $\Tcost_T(m; l_1,n_1,n_2,l_2,b)$.

\begin{thm}
   \label{thm:main-medium}
   Let $C, C' > 0$ be constants,
   and let $M, M'$ be multiplication machines.
   Then there exists a generalised transposition machine $T$
   such that for any input parameters $l_1, l_2, n_1, n_2, b, m$ satisfying
   \begin{equation}
      \label{eq:main-medium-hypothesis}
      \lg(n_1 n_2) < Cb, \qquad \lg^2 b < C' l_1 l_2 n_1 n_2,
      \qquad l_1 l_2 n_1 n_2 b \leq m,
   \end{equation}
   we have
   \[
      \Tcost_T(m; l_1,n_1,n_2,l_2,b)
         = O(\Mcost_M(m) + l_1 l_2 n_1 n_2 \Mcost_{M'}(b)).
   \]
\end{thm}

\begin{proof}
   We are given as input an array $A$ of size
   $l_1 \times n_1 \times n_2 \times l_2$,
   whose entries are $b$-bit strings.
   Let us write $n \coloneqq n_1 n_2$ and $l \coloneqq l_1 l_2$.
   If either $n_1 = 1$ or $n_2 = 1$, then there is nothing to do,
   so we may assume that $n_1, n_2 \geq 2$.

   In the following discussion,
   we frequently deal with array entries whose
   first coordinate is $i_1 \in \{0, \ldots, l_1-1\}$
   and last coordinate is $i_2 \in \{0, \ldots, l_2-1\}$.
   For improved legibility we will place these variables
   in superscripts instead of subscripts.
   For example, we will write $A^{i_1,i_2}_{j_1,j_2}$
   to mean $A_{i_1,j_1,j_2,i_2}$.
   The reader should keep in mind that the arrangement of data on the tape
   is still of course given by the row-major ordering with respect to
   $(i_1, j_1, j_2, i_2)$.

   \step{Step 1: encode input data}.
   We interpret the entries of $A$ as integers in the interval
   $0 \leq A^{i_1,i_2}_{j_1,j_2} < 2^b$.
   Introducing the fixed-point framework of Section \ref{sec:fixed-point},
   with precision
   \[
      p \coloneqq b + \lg(80 n^3) = O(b),
   \]
   we encode these integers into a 3-dimensional array
   $R \in \appdisc^{l_1,n,l_2}$ given by
   \[
      R^{i_1,i_2}_{j_1 n_2 + j_2}
         \coloneqq 2^{-b} A^{i_1,i_2}_{j_1,j_2} \in \appdisc,
         \qquad 0 \leq j_1 < n_1, \;\; 0 \leq j_2 < n_2.
   \]
   This encoding amounts to zero-padding,
   and can be performed in time $O(lnp) = O(m)$.

   \step{Step 2: forward DFTs}.
   For each $(i_1,i_2)$, consider the scaled DFT
   \[
      S^{i_1,i_2}_k \coloneqq
         \frac1n \sum_{j=0}^{n-1} \zeta_n^{-jk} R^{i_1,i_2}_j \in \disc,
         \qquad 0 \leq k < n.
   \]
   According to Proposition \ref{prop:dft-slices},
   we may compute an approximation $\app S \in \disc^{l_1,n,l_2}$
   with $\error(\app S) < 12 n^2$ in time
   $O(\Mcost_M(m') + ln \Mcost_{M'}(p))$
   where $m' \coloneqq l_1 l_2 n_1 n_2 p$.
   Moreover, since $p = O(b)$ and $m' = O(m)$,
   we may use Lemma \ref{lem:chunks} to modify the construction
   in Proposition \ref{prop:dft-slices} slightly,
   replacing all calls to $M'$ and $M$
   with size parameters respectively $p$ and $m'$
   by several invocations with size parameters $b$ and $m$.
   The complexity of computing $\app S$ this way then becomes
   $O(\Mcost_M(m) + ln \Mcost_{M'}(b))$ as desired.
   In the rest of this proof, we will replace $p$ and $m'$
   by $b$ and $m$ in this way without further comment.

   \medskip
   The aim of the remaining steps is to compute the inverse DFTs
   of the slices of~$S$,
   i.e., to recover the coefficients of $R$.
   Recall (see \eqref{eq:inverse-dft}) that these are given by
   \begin{equation}
      \label{eq:inverse-dft-2}
      R^{i_1,i_2}_j = \sum_{k=0}^{n-1} \zeta_n^{jk} S^{i_1,i_2}_k,
      \qquad 0 \leq j < n.
   \end{equation}
   However, we will use a different algorithm to compute the inverse DFT,
   so that the coefficients of $R$ are produced in a different order,
   corresponding to the desired transpose of the original input matrix $A$.

   We begin by writing $j = j_1 n_2 + j_2$ and $k = k_2 n_1 + k_1$
   with $0 \leq j_1, k_1 < n_1$ and $0 \leq j_2, k_2 < n_2$,
   to obtain the well-known Cooley--Tukey decomposition of
   \eqref{eq:inverse-dft-2}:
   \[
      R^{i_1,i_2}_{j_1 n_2 + j_2} =
      \sum_{k_1=0}^{n_1-1}
         \zeta_{n_1}^{j_1 k_1} \cdot \zeta_n^{j_2 k_1}
      \sum_{k_2=0}^{n_2-1}
         \zeta_{n_2}^{j_2 k_2} S^{i_1,i_2}_{k_2 n_1 + k_1}.
   \]
   Let us define intermediate quantities
   $U, V, W \in \disc^{l_1,n_2,n_1,l_2}$ by
   \begin{align*}
      U^{i_1,i_2}_{j_2,k_1} & \coloneqq \frac{1}{n_2} \sum_{k_2=0}^{n_2-1}
         \zeta_{n_2}^{j_2 k_2} S^{i_1,i_2}_{k_2 n_1 + k_1} \in \disc,
         & 0 \leq j_2 < n_2, \;\; 0 \leq k_1 < n_1, \\
      V^{i_1,i_2}_{j_2,k_1} & \coloneqq
         \zeta_n^{j_2 k_1} U^{i_1,i_2}_{j_2,k_1} \in \disc,
         & 0 \leq j_2 < n_2, \;\; 0 \leq k_1 < n_1, \\
      W^{i_1,i_2}_{j_2,j_1} & \coloneqq \frac{1}{n_1} \sum_{k_1=0}^{n_1-1}
         \zeta_{n_1}^{j_1 k_1} V^{i_1,i_2}_{j_2,k_1} \in \disc,
         & 0 \leq j_2 < n_2, \;\; 0 \leq j_1 < n_1,
   \end{align*}
   so that
   \[
      R^{i_1,i_2}_{j_1 n_2 + j_2} = n W^{i_1,i_2}_{j_2,j_1}.
   \]
   In the next three steps, we compute approximations for
   $U$, $V$ and $W$ in turn.

   \step{Step 3: inverse DFTs in $n_2$ direction}.
   Reinterpreting $S$ as an array of size
   $l_1 \times n_2 \times n_1 \times l_2$,
   i.e., as an element of $\disc^{l_1,n_2,n_1,l_2}$,
   the above formula for $U^{i_1,i_2}_{j_2,k_1}$
   amounts to transforming slices of $S$ in the $n_2$ direction.
   Therefore, by Proposition \ref{prop:dft-slices} and Lemma \ref{lem:chunks}
   (see also Remark \ref{rem:combine}),
   we may compute an approximation $\app U \in \appdisc^{l_1,n_2,n_1,l_2}$
   for~$U$ with $\error(\app U) < \error(\app S) + 12 n_2^2 < 24n^2$
   in time $O(\Mcost_M(m) + ln \Mcost_{M'}(b))$.

   \step{Step 4: apply twiddle factors}.
   Let $\gamma_{j_2,k_1} \coloneqq \zeta_n^{j_2 k_1}$.
   Following the method of the proof of Proposition \ref{prop:dft},
   we first compute approximations for $\gamma_{1,k_1} = \zeta_n^{k_1}$
   for $k_1 = 0, \ldots, n_1-1$ with $\error(\app\gamma_{1,k_1}) < 4n_1 - 6$.
   We then similarly compute approximations for
   $\gamma_{j_2,k_1} = (\gamma_{1,k_1})^{j_2}$, for all $k_1$ and $j_2$,
   with $\error(\app\gamma_{j_2,k_1}) < 4(n_1-1)(n_2-1) - 2$.
   Finally we use Lemma \ref{lem:mul} and (Lemma \ref{lem:chunks})
   to compute approximations for
   $V^{i_1,i_2}_{j_2,k_1} = \gamma_{j_2,k_1} U^{i_1,i_2}_{j_2,k_1}$,
   with $\error(\app V) < 4(n_1-1)(n_2-1) - 2 + \error(\app U) + 2 < 28n^2$.
   Altogether this requires time $O(ln \Mcost_{M'}(b))$.
   
   \step{Step 5: inverse DFTs in $n_1$ direction}.
   The definition of $W^{i_1,i_2}_{j_2,j_1}$ amounts to transforming
   the $V$ array in the $n_1$ direction.
   Again by Proposition \ref{prop:dft-slices} we may compute an approximation
   $\app W \in \appdisc^{l_1,n_2,n_1,l_2}$ with
   $\error(\app W) < \error(\app V) + 12n_1^2 < 40n^2$
   in time $O(\Mcost_M(m) + ln \Mcost_{M'}(b))$.

   \step{Step 6: scale and round}.
   The above estimates show that
   \begin{multline*}
      |2^b n \app W^{i_1,i_2}_{j_2,j_1} - A^{i_1,i_2}_{j_1,j_2}|
      = 2^b |n \app W^{i_1,i_2}_{j_2,j_1} - R^{i_1,i_2}_{j_1 n_2 + j_2}| \\
      = 2^b n |\app W^{i_1,i_2}_{j_2,j_1} - W^{i_1,i_2}_{j_2,j_1}|
      \leq 2^{b-p} n \cdot \error(\app W)
      < 2^{b-p} \cdot 40 n^3 \leq \tfrac12.
   \end{multline*}
   In other words, the nearest integer to
   $2^b n \app W^{i_1,i_2}_{j_2,j_1}$
   is exactly $A^{i_1,i_2}_{j_1,j_2}$.
   We may therefore recover $A^{i_1,i_2}_{j_1,j_2}$ by first multiplying
   $2^p \app W^{i_1,i_2}_{j_2,j_1} \in \ZZ[\ii]$ by $n$,
   and then dividing by $2^{p-b}$ and rounding to the nearest integer
   (the imaginary part may be discarded).
   The cost is $O(\Mcost_{M'}(b))$ per coefficient,
   and we have finally succeeded in transposing the original $A$ array.
\end{proof}

\section{Decomposing transpositions into subproblems}
\label{sec:decompose}

In this section we describe a method for decomposing
an $n_1 \times n_2$ transposition problem into smaller problems,
by cutting $n_1$ and/or $n_2$ into chunks.
We use the following terminology:
a \emph{problem of type $(l_1, n_1, n_2, l_2; b)$}
is a generalised transposition problem (see Section \ref{sec:main})
of size $l_1 \times n_1 \times n_2 \times l_2$ with $b$-bit coefficients.

\begin{lem}
   \label{lem:decomposition}
   Let $n_1, n_2, n'_1, n'_2, b \geq 1$,
   and assume that $n'_1 \divides n_1$ and $n'_2 \divides n_2$.
   Then an $n_1 \times n_2$ transposition problem
   with $b$-bit entries may be decomposed into:
   \begin{enumerate}
   \item a problem of type
   $\displaystyle \left(
      \frac{n_1}{n'_1}, n'_1, \frac{n_2}{n'_2}, n'_2; b\right)$,
   followed by
   \item a problem of type
   $\displaystyle \left(
      \frac{n_1 n_2}{n'_1 n'_2}, n'_1, n'_2, 1; b\right)$,
   followed by
   \item a problem of type
   $\displaystyle \left(1, \frac{n_1}{n'_1}, n_2, n'_1; b\right)$.
   \end{enumerate}
\end{lem}
\begin{proof}
   The decomposition is illustrated in Figure \ref{fig:decomposition}.
   We begin with an $n_1 \times n_2$ matrix $A$,
   where the entry $A_{i_1,i_2}$ is stored at index $i_1 n_2 + i_2$.
   Let
   \begin{align*}
      i_1 & = j_1 n'_1 + k_1, \qquad
         0 \leq k_1 < n'_1, \quad 0 \leq j_1 < n_1/n'_1, \\
      i_2 & = j_2 n'_2 + k_2, \qquad
         0 \leq k_2 < n'_2, \quad 0 \leq j_2 < n_2/n'_2.
   \end{align*}
   Then the index $i_1 n_2 + i_2$ may be rewritten as
   \[
      j_1 (n'_1 n_2) + k_1 n_2 + j_2 n'_2 + k_2.
   \]

   In step (1), we exchange the $k_1$ and $j_2$ coordinates.
   After this, the $A_{i_1,i_2}$ entry is stored at index
   \[
      j_1 (n'_1 n_2) + j_2 (n'_1 n_2') + k_1 n'_2 + k_2.
   \]
   
   In step (2), we exchange the $k_1$ and $k_2$ coordinates.
   After this, the $A_{i_1,i_2}$ entry is stored at index
   \begin{align*}
      & \mathrel{\phantom{=}}
         j_1 (n'_1 n_2) + j_2 (n'_1 n_2') + k_2 n'_1 + k_1 \\
      & = j_1 (n'_1 n_2) + i_2 n'_1 + k_1.
   \end{align*}
   
   In step (3), we exchange the $j_1$ and $i_2$ coordinates.
   After this, the $A_{i_1,i_2}$ entry is stored at index
   \begin{align*}
      & \mathrel{\phantom{=}}
         i_2 n_1 + j_1 n'_1 + k_1 \\
      & = i_2 n_1 + i_1,
   \end{align*}
   and we have succeeded in transposing the array.
\end{proof}

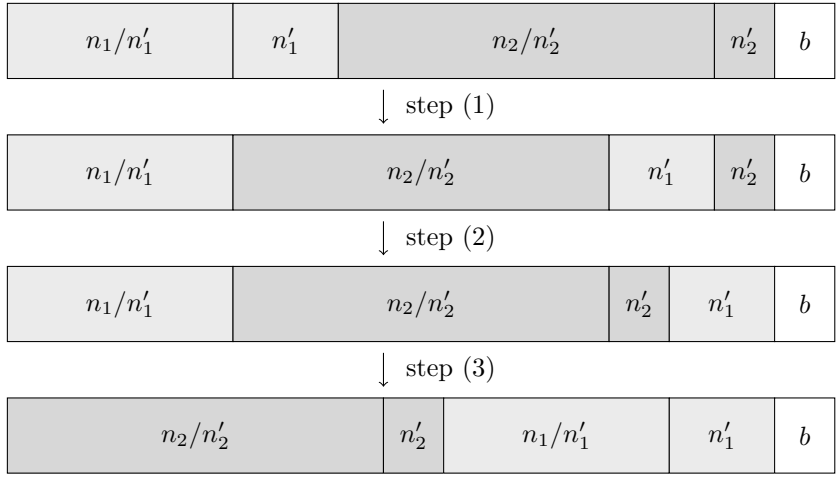
\begin{figure}
   
   \newcommand{\rowheight}{1}
   \newcommand{\rowspace}{1.75}
   \newcommand{\totalheight}{8}
   \newcommand{\arrowgap}{0.65}
   
   \newcommand{\nFracOneWidth}{3}
   \newcommand{\nOneWidth}{1.4}
   \newcommand{\nFracTwoWidth}{5}
   \newcommand{\nTwoWidth}{0.8}
   \newcommand{\bWidth}{0.8}
   
   \pgfmathsetmacro{\totalwidth}%
      {\nFracOneWidth + \nFracTwoWidth + \nOneWidth + \nTwoWidth + \bWidth}

   \begin{tikzpicture}[
    box/.style={rectangle, draw, minimum height=\rowheight cm},
    arrow/.style={->, thick}]

   \definecolor{colourOne}{RGB}{235,235,235}
   \definecolor{colourTwo}{RGB}{215,215,215}
   
   \pgfmathsetmacro{\rowOne}{\totalheight}
   \pgfmathsetmacro{\rowTwo}{\totalheight - \rowspace}
   \pgfmathsetmacro{\rowThree}{\totalheight - 2*\rowspace}
   \pgfmathsetmacro{\rowFour}{\totalheight - 3*\rowspace}
   
   \newcommand{\currentx}{0}
   \newcommand{\newrow}{\pgfmathsetmacro{\currentx}{0}}
   
   \newcommand{\placeFracOne}[1]{%
      \node[box, fill=colourOne, minimum width=\nFracOneWidth cm]
         at (\currentx + \nFracOneWidth/2, #1) {$n_1/n'_1$};%
      \pgfmathsetmacro{\currentx}{\currentx + \nFracOneWidth}%
   }
   \newcommand{\placeOne}[1]{%
      \node[box, fill=colourOne, minimum width=\nOneWidth cm]
         at (\currentx + \nOneWidth/2, #1) {$n'_1$};%
      \pgfmathsetmacro{\currentx}{\currentx + \nOneWidth}%
   }
   \newcommand{\placeFracTwo}[1]{%
      \node[box, fill=colourTwo, minimum width=\nFracTwoWidth cm]
         at (\currentx + \nFracTwoWidth/2, #1) {$n_2/n'_2$};%
      \pgfmathsetmacro{\currentx}{\currentx + \nFracTwoWidth}%
   }
   \newcommand{\placeTwo}[1]{%
      \node[box, fill=colourTwo, minimum width=\nTwoWidth cm]
         at (\currentx + \nTwoWidth/2, #1) {$n'_2$};%
      \pgfmathsetmacro{\currentx}{\currentx + \nTwoWidth}%
   }
   \newcommand{\placeB}[1]{%
      \node[box, fill=white, minimum width=\bWidth cm]
         at (\currentx + \bWidth/2, #1) {$b$};%
      \pgfmathsetmacro{\currentx}{\currentx + \bWidth}%
   }
   
   \newrow
   \placeFracOne{\rowOne}
   \placeOne{\rowOne}
   \placeFracTwo{\rowOne}
   \placeTwo{\rowOne}
   \placeB{\rowOne}
   
   \newrow
   \placeFracOne{\rowTwo}
   \placeFracTwo{\rowTwo}
   \placeOne{\rowTwo}
   \placeTwo{\rowTwo}
   \placeB{\rowTwo}
   
   \newrow
   \placeFracOne{\rowThree}
   \placeFracTwo{\rowThree}
   \placeTwo{\rowThree}
   \placeOne{\rowThree}
   \placeB{\rowThree}
   
   \newrow
   \placeFracTwo{\rowFour}
   \placeTwo{\rowFour}
   \placeFracOne{\rowFour}
   \placeOne{\rowFour}
   \placeB{\rowFour}
   
   \newcommand{\arrowX}{\totalwidth/2 - 0.5}
   \draw[arrow,thin]
      (\arrowX,\rowOne-\arrowgap) -- (\arrowX,\rowTwo+\arrowgap)
      node[right=5pt,midway] {step (1)};
   \draw[arrow,thin]
      (\arrowX,\rowTwo-\arrowgap) -- (\arrowX,\rowThree+\arrowgap)
      node[right=5pt,midway] {step (2)};
   \draw[arrow,thin]
      (\arrowX,\rowThree-\arrowgap) -- (\arrowX,\rowFour+\arrowgap)
      node[right=5pt,midway] {step (3)};
   \end{tikzpicture}
   \caption{Decomposing a transposition (Lemma \ref{lem:decomposition}).}
   \label{fig:decomposition}
\end{figure}

The next result analyses a Turing machine implementation
of the decomposition of Lemma \ref{lem:decomposition}.
We focus on a special case in which $n'_1$ and $n'_2$
are taken to be exponentially smaller than $\max(n_1, n_2)$.
The basic idea will be to reduce steps (1) and (3)
(the ``outer transpositions'') to integer multiplication problems
via Theorem \ref{thm:main-medium},
and to delegate step (2) (the ``inner transposition'')
to some other transposition machine $T'$,
which will be specified later.

For technical reasons it will be convenient to work with power-of-two
parameter sizes.
A \emph{dyadic transposition machine} is a machine having the
same interface as an ordinary transposition machine,
but with the additional restriction that the parameters $n_1, n_2, b$
must lie in $2^\NN \coloneqq \{1, 2, 4, \ldots\}$.

\begin{prop}
   \label{prop:one-step}
   Let $M, M'$ be multiplication machines,
   and let $T'$ be a dyadic transposition machine.
   Then there exists a dyadic transposition machine $T$
   with the following properties.
   
   Let $n_1, n_2, b \in 2^\NN$ and $m \geq 1$,
   and assume that $n_1 n_2 b \leq m$.
   Let
   \[
      s \coloneqq 2^{\lg \lg \max(n_1, n_2)},
   \]
   and define
   \[
      n'_1 \coloneqq \min(n_1, s), \qquad n'_2 \coloneqq \min(n_2, s),
      \qquad m' \coloneqq n'_1 n'_2 b.
   \]
   Then
   \[
      \Tcost_T(m; n_1, n_2, b) <
         \frac{n_1 n_2}{n'_1 n'_2} \Tcost_{T'}(m'; n'_1, n'_2, b)
         + O\left(\Mcost_M(m) + \frac{m}{s} \Mcost_{M'}(s)\right).
   \]
\end{prop}
\begin{proof}
   Clearly $n'_1 \divides n_1$, since $n'_1 \leq n_1$ and $n'_1 \in 2^\NN$.
   Similarly $n'_2 \divides n_2$.
   Let us now consider in turn each step from Lemma \ref{lem:decomposition}.

   \medskip
   \emph{Step (1)}.
   If $s \geq n_2$ then $n'_2 = n_2$,
   so this step becomes trivial.
   Assume therefore that $s < n_2$,
   so that $n'_2 = s$,
   and we are dealing with a problem of type
   \[
      \left(\frac{n_1}{n'_1}, n'_1, \frac{n_2}{s}, s; b\right).
   \]
   This is equivalent to a problem of type
   \[
      \left(\frac{n_1}{n'_1}, n'_1, \frac{n_2}{s}, b; s\right),
   \]
   i.e., by simply reinterpreting each length-$s$ array of $b$-bit entries
   as a length-$b$ array of $s$-bit entries.
   We will effect this transposition by invoking
   Theorem \ref{thm:main-medium} with the parameters
   \[
      (l_1, n_1, n_2, l_2, b) \coloneqq
         \left(\frac{n_1}{n'_1}, n'_1, \frac{n_2}{s}, b, s\right)
   \]
   (and $m \coloneqq m$, $M \coloneqq M$, $M' \coloneqq M'$).
   Let us verify the hypotheses \eqref{eq:main-medium-hypothesis}
   for these parameters.
   The first inequality states that $\lg(n'_1 (n_2/s)) < Cs$
   for a suitable constant $C > 0$.
   This holds as
   \[
      \lg\left(n'_1 \frac{n_2}{s}\right)
         \leq \lg(n_1 n_2) \leq 2 \lg \max(n_1, n_2) \leq 2s.
   \]
   For the second inequality,
   we must show that
   \[
      \lg^2 s < C' \frac{n_1}{n'_1} n'_1 \frac{n_2}{s} b = C' n_1 n_2 b / s
   \]
   for suitable $C'$.
   Writing $r \coloneqq \max(n_1, n_2)$,
   we have $s = 2^{\lg \lg r} \leq 2 \lg r$ and $\lg s = \lg \lg r$,
   so
   \[
      s \lg^2 s \leq (2 \lg r) (\lg \lg r)^2 < C' r \leq C' n_1 n_2 b
   \]
   for large enough $C'$ (in fact $C' = 6$ will do).
   The third inequality is clear.
   According to Theorem \ref{thm:main-medium},
   the cost of this transposition is
   \[
      O\left(\Mcost_M(m) +
         \frac{n_1}{n'_1} n'_1 \frac{n_2}{s} b \Mcost_{M'}(s)\right)
      = O\left(\Mcost_M(m) + \frac{m}{s} \Mcost_{M'}(s)\right).
   \]

   \medskip
   \emph{Step (2)}.
   This problem is equivalent to
   $n_1 n_2/n'_1 n'_2$ separate two-dimensional
   $n'_1 \times n'_2$ transposition problems with $b$-bit coefficients.
   Using $T'$, the cost is
   \[
      \frac{n_1 n_2}{n'_1 n'_2} \Tcost_{T'}(m'; n'_1, n'_2, b).
   \]

   \medskip
   \emph{Step (3)}.
   If $s \geq n_1$ then $n'_1 = n_1$ and the problem is trivial,
   so assume that $s < n_1$.
   Then $n'_1 = s$ and we face a problem of type
   \[
      \left(1, \frac{n_1}{s}, n_2, s; b\right).
   \]
   As in step (1), we may reinterpret this as a problem of type
   \[
      \left(1, \frac{n_1}{s}, n_2, b; s\right).
   \]
   Again we invoke Theorem \ref{thm:main-medium},
   this time with parameters
   \[
      (l_1, n_1, n_2, l_2, b) \coloneqq
         \left(1, \frac{n_1}{s}, n_2, b, s\right).
   \]
   To see that the hypotheses \eqref{eq:main-medium-hypothesis}
   are satisfied, observe that
   \[
      \lg\left(\frac{n_1}{s} n_2 \right) \leq \lg(n_1 n_2) \leq 2s,
   \]
   and by the same argument as before,
   \[
      \lg^2 s < C' n_1 n_2 b / s = C' \frac{n_1}{s} n_2 b.
   \]
   The cost of this transposition is therefore given by
   \[
      O\left(\Mcost_M(m) +
         \frac{n_1}{s} n_2 b \Mcost_{M'}(s)\right)
      = O\left(\Mcost_M(m) + \frac{m}{s} \Mcost_{M'}(s)\right).
      \qedhere
   \]
\end{proof}

\section{Proofs of main results}
\label{sec:proofs}

In this section we prove the main results stated in Section \ref{sec:summary}.

\subsection{Proof of Theorem \ref{thm:main-implication}}

We first establish the following reduction
from transposition to multiplication.

\begin{prop}
   \label{prop:main-dyadic}
   Let $M$ be a multiplication machine.
   Then there exists a dyadic transposition machine $T$ such that
   \[
      \Tcost_T(m; n_1, n_2, b) = O(\Mcost_M(m) + m \lg \lg \max(n_1, n_2)).
   \]
\end{prop}
\begin{proof}
   We apply Proposition \ref{prop:one-step}
   with $M' \coloneqq M^*$ and $T' \coloneqq T^*$.
   By \eqref{eq:Mstar-bound} and \eqref{eq:Tstar-bound},
   the resulting dyadic transposition machine $T$ satisfies
   \begin{align*}
      \Tcost_T(m; n_1, n_2, b)
         & < \frac{n_1 n_2}{n'_1 n'_2} \Tcost_{T^*}(m'; n'_1, n'_2, b)
            + O\left(\Mcost_M(m) + \frac{m}{s} \Mcost_{M^*}(s) \right) \\
         & = O\left( n_1 n_2 b \lg \min(n'_1, n'_2) + \Mcost_M(m)
               + m \lg s \right),
   \end{align*}
   where $n'_1, n'_2, s, m'$ are defined as in Proposition \ref{prop:one-step}.
   Since $n'_i \leq s$, $n_1 n_2 b \leq m$
   and $\lg s = \lg \lg \max(n_1, n_2)$,
   \[
      \Tcost_T(m; n_1, n_2, b)
         = O(\Mcost_M(m) + m \lg \lg \max(n_1, n_2)).
      \qedhere
   \]
\end{proof}

It is straightforward to remove the dyadic restriction:
\begin{cor}
   \label{cor:main-drop-dyadic}
   Let $M$ be a multiplication machine.
   Then there exists an (ordinary non-dyadic)
   transposition machine $T$ such that
   \[
      \Tcost_T(m; n_1, n_2, b) = O(\Mcost_M(m) + m \lg \lg \max(n_1, n_2)).
   \]
\end{cor}
\begin{proof}
   Let $n_1, n_2, b, m \geq 1$ be arbitrary input parameters
   such that $n_1 n_2 b \leq m$.
   We round up the parameters to powers of two, setting
   \[
      \tilde n_1 \coloneqq 2^{\lceil \log_2 n_1 \rceil}, \qquad
      \tilde n_2 \coloneqq 2^{\lceil \log_2 n_2 \rceil}, \qquad
      \tilde b \coloneqq 2^{\lceil \log_2 b \rceil}.
   \]
   Let $\tilde m \coloneqq 8m$ so that
   $\tilde n_1 \tilde n_2 \tilde b \leq \tilde m$.

   We first copy the $n_1 \times n_2$ input matrix with $b$-bit entries
   into an $\tilde n_1 \times \tilde n_2$ matrix with $\tilde b$-bit entries,
   by suitable zero-padding.
   Clearly this can be done in time $O(m)$.
   We then invoke the machine from Proposition \ref{prop:main-dyadic}
   with parameters
   $\tilde n_1, \tilde n_2, \tilde b, \tilde m$ to transpose this matrix.
   This requires time
   \[
      O(\Mcost_M(\tilde m) + \tilde m \lg \lg \max(\tilde n_1, \tilde n_2))
      = O(\Mcost_M(8m) + m \lg \lg \max(n_1, n_2)).
   \]
   Finally we copy everything back to the desired
   $n_2 \times n_1$ output matrix
   (with $b$-bit entries) in time $O(m)$.

   To get the desired complexity bound, we make one more small change:
   using Lemma \ref{lem:chunks}, we replace all calls to $M$ with
   size parameter $8m$ by several calls with $m$ instead.
   The resulting machine $T$ satisfies
   \[
      \Tcost_T(m; n_1, n_2, b) = O(\Mcost_M(m) + m \lg \lg \max(n_1, n_2)).
      \qedhere
   \]
\end{proof}

Specialising to the binary case,
and restating in terms of lower bounds, we obtain the following.
\begin{cor}
   \label{cor:main}
   Assume that
   \[
      \Tcost_T(m) = \omega(m \lg \lg m)
   \]
   for every binary transposition machine $T$.
   Then for any multiplication machine $M$,
   there exists a binary transposition machine $T$ such that
   \[
      \Mcost_M(m) = \Omega(\Tcost_T(m)).
   \]
\end{cor}
\begin{proof}
   Taking $n_1, n_2 \coloneqq \lfloor m^{1/2} \rfloor$
   and $b \coloneqq 1$ in Corollary \ref{cor:main-drop-dyadic},
   we get a binary transposition machine satisfying
   \[
      \Tcost_T(m) = O(\Mcost_M(m) + m \lg \lg m).
   \]
   Therefore, there is an absolute constant $C > 0$ such that
   \[
      \Tcost_T(m) < C (\Mcost_M(m) + m \lg \lg m)
   \]
   for all $m \geq 1$, and hence
   \[
      \Mcost_M(m) > \frac{\Tcost_T(m)}{C}  - m \lg \lg m
      = \left(\frac{1}{C} - \frac{m \lg \lg m}{\Tcost_T(m)}\right) \Tcost_T(m)
   \]
   for all $m \geq 1$.
   Since we assumed that $\Tcost_T(m) = \omega(m \lg \lg m)$,
   we conclude that in fact $\Mcost_M(m) = \Omega(\Tcost_T(m))$.
\end{proof}

Theorem \ref{thm:main-implication} now follows immediately
from Corollary \ref{cor:main}.

\begin{rem}
   The $m \lg \lg m$ term above arises from the use of the folklore
   transposition algorithm to handle the ``inner'' transpositions
   of size exponentially smaller than $m$.
   In the following sections we will iterate the decomposition strategy
   from Section \ref{sec:decompose},
   in an attempt to make the innermost transpositions even smaller,
   and hence further reduce the size of the $m \lg \lg m$ term.
\end{rem}

\subsection{Proof of Theorem \ref{thm:main-ell-implication}}

For any multiplication machine~$M$, define
\[
   \Mcostinc_M(m) \coloneqq
      m \cdot \max_{k \leq m} \frac{\Mcost_M(k)}{k}, \qquad m \geq 1.
\]
Observe that by construction, $\Mcostinc_M(m)/m$ is non-decreasing in $m$.

\begin{rem}
   If $\Mcost_M(m)/m$ is ``non-decreasing up to a constant'',
   i.e., there exists a constant $C \geq 1$ such that
   $\Mcost_M(m')/m' \leq C \Mcost_M(m)/m$ for all $m' \leq m$,
   which is the case for all ``reasonable'' multiplication algorithms
   known to the authors,
   then $\Mcostinc_M(m) = O(\Mcost_M(m))$.
\end{rem}

We now have the following variant of Proposition \ref{prop:main-dyadic},
which has a smaller error term,
but at the cost of replacing $\Mcost(m)$ by $\Mcostinc(m)$.
\begin{prop}
   \label{prop:main-ell-dyadic}
   Fix an integer $\ell \geq 2$,
   and let $M$ be a multiplication machine.
   Then there exists a dyadic transposition machine $T$ such that
   \[
      \Tcost_T(m; n_1, n_2, b)
         = O(\Mcostinc_M(m) + m \lg^{\circ \ell} \max(n_1, n_2)).
   \]
\end{prop}

\newcommand{\fixalign}[1]{#1_{\phantom{0}}}

\begin{proof}
   The idea is to chain together $\ell-1$ invocations of
   Proposition \ref{prop:one-step}.
   We start by defining $T_\ell \coloneqq T^*$.
   Then for $j = \ell-1, \ldots, 1$ in turn,
   we apply Proposition \ref{prop:one-step} with $T' \coloneqq T_{j+1}$
   and $M' \coloneqq M$ to obtain a new machine $T_j$.
   We claim that the final top-level machine $T_1$
   satisfies the desired bound
   \[
      \Tcost_{T_1}(m; n_1, n_2, b)
         = O(\Mcostinc_M(m) + m \lg^{\circ \ell} \max(n_1, n_2)).
   \]
   
   To prove this,
   let $n_1, n_2, b \in 2^\NN$ and $m \geq 1$ be the input parameters
   for $T_1$, with $n_1 n_2 b \leq m$.
   For $j = 1, \ldots, \ell$, let $n_1^{(j)}$, $n_2^{(j)}$,
   $\fixalign{m}^{(j)}$ be the parameters subsequently passed
   to $T_j$ as we proceed down the call chain;
   in particular, $n_i^{(1)} = n_i$.
   For $j = 1, \ldots, \ell-1$ let $s^{(j)}$ be the
   corresponding value of $s$ as in the statement of
   Proposition \ref{prop:one-step}.

   Let us show that
   \begin{equation}
      \label{eq:logmax}
      \lg \max(n_1^{(j)}, n_2^{(j)}) = \lg^{\circ j} \max(n_1, n_2),
         \qquad j = 1, \ldots, \ell.
   \end{equation}
   Suppose first that $n_1 \geq n_2$.
   It follows that $n_1^{(j)} \geq n_2^{(j)}$ for all $j = 1, \ldots, \ell$,
   since by induction we have
   $n_1^{(j+1)} = \min(n_1^{(j)}, \fixalign{s}^{(j)})
   \geq \min(n_2^{(j)}, \fixalign{s}^{(j)}) = n_2^{(j+1)}$.
   Hence
   \[
      s^{(j)} = 2^{\lg \lg \max(n_1^{(j)}, n_2^{(j)})}
         = 2^{\lg \lg n_1^{(j)}}
   \]
   and
   \[
      \lg n_1^{(j+1)} = \lg \min(n_1^{(j)}, \fixalign{s}^{(j)})
         = \min(\lg n_1^{(j)}, \lg \lg n_1^{(j)}) = \lg \lg n_1^{(j)}
   \]
   for $j = 1, \ldots, \ell-1$.
   By induction $\lg n_1^{(j)} = \lg^{\circ j} n_1$
   for $j = 1, \ldots, \ell$,
   so \eqref{eq:logmax} holds.
   If $n_2 \geq n_1$, a symmetrical argument applies.
      
   Now we estimate the complexity.
   For each $j = 1, \ldots, \ell-1$ we have
   \begin{multline*}
      \Tcost_{T_j}\left(\fixalign{m}^{(j)}; n_1^{(j)}, n_2^{(j)}, b\right) <
         \frac{n_1^{(j)} n_2^{(j)}}{n_1^{(j+1)} n_2^{(j+1)}}
         \Tcost_{T_{j+1}}\left(\fixalign{m}^{(j+1)};
            n_1^{(j+1)}, n_2^{(j+1)}, b\right) \\
         + O\left(\Mcost_M(m^{(j)}) +
            \frac{m^{(j)}}{s^{(j)}} \Mcost_M(s^{(j)})\right).
   \end{multline*}
   Combining these inequalities, we obtain
   \begin{multline*}
      \Tcost_{T_1}(m; n_1, n_2, b) <
         \frac{n_1 n_2}{n_1^{(\ell)} n_2^{(\ell)}}
         \Tcost_{T_\ell}\left(\fixalign{m}^{(\ell)};
                              n_1^{(\ell)}, n_2^{(\ell)}, b\right) \\
         + \sum_{j=1}^{\ell-1} O\left(\frac{n_1 n_2}{n_1^{(j)} n_2^{(j)}}
            \left( \Mcost_M(m^{(j)}) +
               \frac{m^{(j)}}{s^{(j)}} \Mcost_M(s^{(j)})\right)\right).
   \end{multline*}
   Since $T_\ell = T^*$, by \eqref{eq:logmax} the first term becomes
   \begin{align*}
      \frac{n_1 n_2}{n_1^{(\ell)} n_2^{(\ell)}} \cdot
               O\left( n_1^{(\ell)} n_2^{(\ell)} b
                  \lg \min(n_1^{(\ell)}, n_2^{(\ell)}) \right)
         & = O(m \lg \max(n_1^{(\ell)}, n_2^{(\ell)})) \\
         & = O(m \lg^{\circ \ell} \max(n_1, n_2)).
   \end{align*}

   For the second term, first observe that
   \[
      \frac{n_1 n_2}{n_1^{(j)} n_2^{(j)}} \leq \frac{m}{m^{(j)}},
      \qquad j = 1, \ldots, \ell-1.
   \]
   Indeed, for $j = 1$ the inequality holds as both sides are equal to $1$,
   and for $j > 1$ it holds as $\fixalign{m}^{(j)} = n_1^{(j)} n_2^{(j)} b$
   and $n_1 n_2 b \leq m$.
   Therefore the second term becomes
   \begin{equation}
      \label{eq:second-term-bound}
      \sum_{j=1}^{\ell-1} O\left( m \frac{\Mcost_M(m^{(j)})}{m^{(j)}}
         + m \frac{\Mcost_M(s^{(j)})}{s^{(j)}} \right).
   \end{equation}

   We now claim that $m^{(j)} \leq m$ and $s^{(j)} \leq m$ for all $j < \ell$.
   The first inequality is trivial for $j = 1$,
   and for $j > 1$ we have
   $\fixalign{m}^{(j)} = n_1^{(j)} n_2^{(j)} b \leq n_1 n_2 b \leq m$
   since clearly $n_i^{(j)} \leq n_i$.
   For the second inequality, since $2^{\lg \lg n} \leq n$ for all $n \geq 2$,
   \[
      s^{(j)} = 2^{\lg \lg \max(n_1^{(j)}, n_2^{(j)})} \leq 2^{\lg \lg \max(n_1, n_2)} \leq 2^{\lg \lg(n_1 n_2)} \leq n_1 n_2 \leq m
   \]
   for $j < \ell$.
   (Strictly speaking this argument only works for $n_1 n_2 \geq 2$;
   we ignore the trivial case $n_1 = n_2 = 1$.)

   Therefore \eqref{eq:second-term-bound} becomes
   $\sum_{j=1}^{\ell-1} O(\Mcostinc_M(m))$,
   and since $\ell$ is fixed,
   this simplifies further to $O(\Mcostinc_M(m))$.
   We conclude that
   \[
      \Tcost_{T_1}(m; n_1, n_2, b)
         = O(\Mcostinc_M(m) + m \lg^{\circ \ell} \max(n_1, n_2)). \qedhere
   \]
\end{proof}

Again it is easy to drop the dyadic restriction:
\begin{cor}
   \label{cor:main-ell-drop-dyadic}
   Fix an integer $\ell \geq 2$,
   and let $M$ be a multiplication machine.
   Then there exists an (ordinary, non-dyadic)
   transposition machine $T$ such that
   \[
      \Tcost_T(m; n_1, n_2, b)
         = O(\Mcostinc_M(m) + m \lg^{\circ \ell} \max(n_1, n_2)).
   \]
\end{cor}
\begin{proof}
   As in the proof of Corollary \ref{cor:main-drop-dyadic},
   we use a zero-padding strategy to reduce to the dyadic case
   (Proposition \ref{prop:main-ell-dyadic}),
   and then we use Lemma \ref{lem:chunks} to replace every call to $M$
   with size parameter $\geq m$ by $O(1)$ calls with size $m$.
   (To be completely pedantic,
   this requires modifying the transposition machine interface slightly
   to enable the top-level parameter $m$ to be passed down the call chain,
   at least for the top few levels.
   This does not affect the complexity.)
\end{proof}
\begin{proof}[Proof of Theorem \ref{thm:main-ell-implication}]
   Suppose that we are given a multiplication machine $M$
   such that $\Mcost_M(m) = O(m \lg^{\circ \ell} m)$.
   Then
   \begin{equation}
      \label{eq:Mcostinc-bound}
      \Mcostinc_M(m) = m \cdot \max_{k \leq m} \frac{\Mcost_M(k)}{k}
         = m \cdot \max_{k \leq m} O(\lg^{\circ \ell} k)
         = O(m \lg^{\circ \ell} m).
   \end{equation}
   Taking $n_1, n_2 \coloneqq \lfloor m^{1/2} \rfloor$
   and $b \coloneqq 1$ in Corollary \ref{cor:main-ell-drop-dyadic},
   we get a binary transposition machine $T$ satisfying
   $\Tcost_T(m) = O(m \lg^{\circ \ell} m)$.
   But this contradicts the hypothesis that no such machine exists.
\end{proof}

\subsection{Proof of Theorem \ref{thm:main-recursive-implication}}

The following reduction may be regarded as
a more aggressive version of Proposition \ref{prop:main-ell-dyadic}.
\begin{prop}
   \label{prop:main-recursive-dyadic}
   Let $M$ be a multiplication machine.
   Then there exists a dyadic transposition machine $T$ such that
   \[
      \Tcost_T(m; n_1, n_2, b) = O( \ell_0 \Mcostinc(m)), \qquad
      \ell_0 \coloneqq \max(1, \lg^* \max(n_1, n_2) - \lg^* b).
   \]
\end{prop}
Comparing with Proposition \ref{prop:main-ell-dyadic},
we see that the error term has been completely eliminated,
but we pay with a factor of $\ell_0$ in the main term.

\begin{proof}
   The argument is similar to the proof of Proposition \ref{prop:main-ell-dyadic},
   but instead of constructing a sequence of machines,
   we must construct a single machine that calls itself recursively,
   where the number of recursion levels depends on the relative
   sizes of $\max(n_1, n_2)$ and $b$.
   
   \medskip
   \emph{The base case.}
   We first describe a machine $T'$ that handles
   the base case of the recursion.
   Instead of calling the folklore algorithm $T^*$
   as was done in the proof of Proposition \ref{prop:main-ell-dyadic},
   we will construct $T'$ via Theorem \ref{thm:main-medium}.
   The machine $T'$ has the same interface as a dyadic transposition machine,
   i.e., the input consists of integers $n_1, n_2, b \in 2^\NN$
   and $m \geq 1$ such that $n_1 n_2 b \leq m$,
   but we impose the additional constraint that
   \[
      \lg \max(n_1, n_2) \leq b.
   \]
   
   Let $s \coloneqq 2^{\lg \lg \max(n_1, n_2)}$.
   Then $s$ is a power of two and $\lg s \leq \lg b$.
   Since $b$ is also a power of two, this implies that $s \divides b$.
   We perform the transposition by invoking Theorem \ref{thm:main-medium}
   with $M' \coloneqq M$ and with parameters
   \[
      (l_1, n_1, n_2, l_2, b)
         \coloneqq \left(1, n_1, n_2, \frac{b}{s}, s\right),
   \]
   i.e., we reinterpret the $n_1 \times n_2$ input matrix
   with $b$-bit entries as an array of size
   $1 \times n_1 \times n_2 \times (b/s)$ with $s$-bit entries.
   Let us verify the hypotheses in \eqref{eq:main-medium-hypothesis}.
   The first inequality states that $\lg(n_1 n_2) < Cs$ for suitable $C > 0$;
   this holds as $\lg(n_1 n_2) \leq 2 \lg \max(n_1, n_2) \leq 2s$.
   For the second inequality we must prove that
   $\lg^2 s < C' n_1 n_2 b / s$;
   this follows by a similar argument to step (1) in the proof
   of Proposition \ref{prop:one-step}.
   The cost of this transposition is thus
   \[
      O\left(\Mcost_M(m) + n_1 n_2 \frac{b}{s} \Mcost_M(s)\right)
      = O\left(\Mcost_M(m) + \frac{m}{s} \Mcost_M(s)\right).
   \]
   Since $s \leq b \leq m$, this bound simplifies to $O(\Mcostinc_M(m))$.

   (We remark that it would not have worked to invoke
   Theorem \ref{thm:main-medium} directly for the original
   problem of type $(1, n_1, n_2, 1; b)$,
   because the second inequality in \eqref{eq:main-medium-hypothesis}
   might fail if $b$ is much larger than $n_1$ and $n_2$.
   This is the reason for introducing the auxiliary parameter $s$,
   and indirectly the reason for requiring $b$ to be a power of two.)
   
   \medskip
   \emph{The recursive case.}
   We now describe the main dyadic transposition machine $T$,
   which takes as input any $n_1, n_2, b \in 2^\NN$ and $m \geq 1$
   such that $n_1 n_2 b \leq m$.

   If $\lg \max(n_1, n_2) > b$,
   we use the same strategy as in Proposition \ref{prop:one-step}
   (taking $M' \coloneqq M$),
   i.e., we set
   \[
      s \coloneqq 2^{\lg \lg \max(n_1, n_2)}, \qquad
      n'_i \coloneqq \min(n_i, s), \qquad m' \coloneqq n'_1 n'_2 b,
   \]
   and decompose into $n_1 n_2 / n'_1 n'_2$ transpositions of size
   $n'_1 \times n'_2$, plus additional work of cost
   \[
      O\left(\Mcost_M(m) + \frac{m}{s} \Mcost_M(s)\right).
   \]
   We then call $T$ recursively for each
   small $n'_1 \times n'_2$ subproblem.
   The recursion continues until we encounter
   $\lg \max(n_1, n_2) \leq b$, at which point we finally call
   the base case machine $T'$.
   
   Let us estimate the number of recursive calls,
   and in particular show that this quantity is finite.
   For $j = 1, 2, \ldots$,
   let $n_1^{(j)}$, $n_2^{(j)}$ and $\fixalign{m}^{(j)}$
   denote the values of $n_1$, $n_2$ and $m$ passed to
   the $j$-th recursive call,
   where $j = 1$ corresponds to the initial call.
   As shown in the proof of Proposition \ref{prop:main-ell-dyadic}
   (see \eqref{eq:logmax}), we have
   \[
      \lg \max(n_1^{(j)}, n_2^{(j)}) = \lg^{\circ j} \max(n_1, n_2)
   \]
   for each $j$.
   The sequence $\lg^{\circ j} \max(n_1, n_2)$ is strictly decreasing,
   so we must eventually reach $\lg \max(n_1^{(j)}, n_2^{(j)}) \leq b$.
   Let $\ell$ be the smallest $j$ for which this occurs,
   i.e., the number of recursive calls (including the initial call).

   We claim that
   \begin{equation}
      \label{eq:ell-bound}
      \ell \leq \max(1, \lg^* \max(n_1, n_2) - \lg^* b + 1).
   \end{equation}
   This clearly holds if $\ell = 1$.
   Suppose now that $\ell \geq 2$.
   Then
   \[
      \lg^{\circ(\ell-1)} \max(n_1, n_2) > b.
   \]
   If $b = 1$, this implies that $\lg^* \max(n_1, n_2) > \ell - 1$,
   and $\lg^* b = 0$, so \eqref{eq:ell-bound} holds.
   If $b \geq 2$, then $\lg^* b \geq 1$,
   so applying $\lg^{\circ((\lg^* b) - 1)}$ to both sides,
   \[
      \lg^{\circ(\ell + (\lg^* b) - 2)} \max(n_1, n_2)
         \geq \lg^{\circ((\lg^* b) - 1)} b > 1.
   \]
   This implies that $\lg^* \max(n_1, n_2) > \ell + \lg^* b - 2$,
   i.e., $\ell \leq \lg^* \max(n_1, n_2) - \lg^* b + 1$,
   so again \eqref{eq:ell-bound} holds.

   The same complexity argument as in the proof of
   Proposition \ref{prop:main-ell-dyadic} now shows that
   \[
      \Tcost_T(m; n_1, n_2, b) < \frac{n_1 n_2}{n_1^{(\ell)} n_2^{(\ell)}}
         \Tcost_{T'}(\fixalign{m}^{(\ell)}; n_1^{(\ell)}, n_2^{(\ell)}, b)
         + \sum_{j=1}^{\ell-1} O(\Mcostinc_M(m)).
   \]
   (Note here that the big-$O$ constant is the same at each
   recursion level.)
   We showed above that the cost of each call to the base case is
   $O(\Mcostinc_M(m^{(\ell)}))$,
   so the first term becomes
   \[
      O\left(\frac{n_1 n_2}{n_1^{(\ell)} n_2^{(\ell)}}
         \Mcostinc_M(m^{(\ell)})\right)
      = O\left(\frac{m}{m^{(\ell)}} \Mcostinc_M(m^{(\ell)})\right)
      = O(\Mcostinc_M(m)).
   \]
   Therefore we finally obtain
   \[
      \Tcost_T(m; n_1, n_2, b) =
         O(\Mcostinc_M(m)) + \sum_{j=1}^{\ell-1} O(\Mcostinc_M(m))
         = O(\ell \Mcostinc_M(m)).
      \qedhere
   \]
\end{proof}

As usual, we have the following non-dyadic version.
\begin{cor}
   \label{cor:main-recursive-drop-dyadic}
   Let $M$ be a multiplication machine.
   Then there exists an (ordinary, non-dyadic)
   transposition machine $T$ such that
   \[
      \Tcost_T(m; n_1, n_2, b) = O( \ell_0 \Mcostinc_M(m)), \qquad
      \ell_0 \coloneqq \max(1, \lg^* \max(n_1, n_2) - \lg^* b).
   \]
\end{cor}
\begin{proof}
   Applying the rounding and zero-padding strategy
   of Corollary \ref{cor:main-drop-dyadic},
   we obtain the complexity bound $O(\tilde\ell \Mcostinc_M(\tilde m))$,
   where by \eqref{eq:ell-bound} we have
   \[
      \tilde\ell \leq \max(1,
         \lg^* \max(\tilde n_1, \tilde n_2) - \lg^* \tilde b + 1).
   \]
   Clearly $\lg^* \tilde n_i = \lg^* n_i$ and $\lg^* \tilde b = \lg^* b$,
   so $\tilde \ell = O(\ell_0)$ where $\ell_0$ is defined as in
   the statement of Proposition \ref{prop:main-recursive-dyadic}.
   Finally, modifying the machine via Lemma \ref{lem:chunks}
   in the usual way,
   we may replace $\Mcostinc_M(\tilde m)$ by $\Mcostinc_M(m)$,
   to obtain the desired bound $O(\ell_0 \Mcostinc_M(m))$.
\end{proof}

\begin{proof}[Proof of Theorem \ref{thm:main-recursive-implication}]
   Let $f(m)$ be a non-decreasing function,
   and suppose that there exists a multiplication machine $M$ satisfying
   $\Mcost_M(m) = O(m f(m))$.
   Then
   \[
      \Mcostinc_M(m) = m \cdot \max_{k \leq m} \frac{\Mcost_M(k)}{k}
         = m \cdot \max_{k \leq m} O(f(k))
         = O(m f(m)).
   \]
   Applying Corollary \ref{cor:main-recursive-drop-dyadic} with
   $b \coloneqq 1$ and $n_1, n_2 \coloneqq \lfloor m^{1/2} \rfloor$,
   we obtain a binary transposition machine $T$ such that
   \[
      \Tcost_T(m)
         = O( \lg^*(\lfloor m^{1/2} \rfloor) \Mcostinc_M(m) )
         = O(m f(m) \lg^* m).
   \]
   But this contradicts the hypothesis that no such machine exists.
\end{proof}

\begin{rem}
   Using the $\Mcostinc_M(m)$ notation,
   it is possible to express Theorem \ref{thm:main-recursive-implication}
   directly in the form of a lower bound.
   Namely, the theorem is equivalent to the following statement:
   for any multiplication machine~$M$,
   there exists a binary transposition machine $T$ such that
   \[
      \Mcostinc_M(m)
         = \Omega\bigg(\frac{\Tcost_T(m)}{\lg^* m}\bigg).
   \]
\end{rem}

\subsection{Transposition with larger coefficients}
\label{sec:larger-coeffs}

We mentioned in Section \ref{sec:summary} that our methods are not quite
strong enough to prove the following statement:
if the binary transposition problem cannot be solved in linear time,
then multiplication cannot be carried out in linear time.
To make progress towards this goal,
with our current methods it appears to be necessary to abandon the binary
case and work instead with larger coefficients.
In this section we briefly outline a result in this direction.

Fix an integer $\ell \geq 1$,
and for any $m \geq 1$ consider the transposition problem with parameters
\begin{equation}
   \label{eq:log-machine-parameters}
   b \coloneqq \lg^{\circ \ell} m,
      \qquad n_1, n_2 = n \coloneqq \lfloor (m / b)^{1/2} \rfloor.
\end{equation}
A machine~$T$ performing this type of transposition will be called an
\emph{$\ell$-logarithmic transposition machine}.
We denote its worst-case running time by $\Tcost^\ell_T(m)$.

\begin{thm}
   \label{thm:log-implication}
   Fix $\ell \geq 1$ and a non-decreasing function $f(m)$.
   If no $\ell$-logarithmic transposition machine achieves
   $\Tcost^\ell(m) = O(m f(m))$,
   then no multiplication machine achieves $\Mcost(m) = O(m f(m))$.
\end{thm}

For example,
if $\ell$-logarithmic transposition cannot be performed in linear time
--- even for one fixed value of $\ell$ ---
then multiplication cannot be carried out in linear time.
In this sense, Theorem \ref{thm:log-implication} is slightly
stronger than Theorem \ref{thm:main-recursive-implication},
although it does involve an arguably less natural transposition problem.

\begin{proof}
   Suppose that $M$ is a multiplication machine such that
   $\Mcost_M(m) = O(m f(m))$.
   Apply Corollary \ref{cor:main-recursive-drop-dyadic} with $n_1, n_2, b$
   defined as in \eqref{eq:log-machine-parameters}.
   This yields an $\ell$-logarithmic transposition machine $T$ such that
   $\Tcost^\ell_T(m) = O(\ell_0 \Mcostinc_M(m))$ where
   \[
      \ell_0 = \max(1, \lg^* \max(n_1, n_2) - \lg^* b)
         \leq \max(1, \lg^* m - \lg^* b).
   \]
   But since $b = \lg^{\circ \ell} m$ we clearly have
   $\lg^* m \leq \lg^* b + \ell$.
   Thus $\ell_0 \leq \ell = O(1)$, since $\ell$ was assumed fixed.
   The same calculation as in the proof
   of Theorem \ref{thm:main-recursive-implication}
   then shows that $\Tcost^\ell_T(m) = O(\Mcostinc_M(m)) = O(m f(m))$,
   contradicting the hypothesis.
\end{proof}

\section{Discussion and perspectives}

\subsection{Alternative complexity models}

It is interesting to ask what impact our reductions have
in complexity models other than the multitape Turing machine.

One model often discussed in the context of integer multiplication
is the \emph{Boolean circuit model} \cite[\S4.3]{Pap-complexity}.
The $O(n \log n)$ multiplication algorithm of \cite{HvdH-nlogn}
works in this model,
and it is reasonable to guess that this is optimal.
Unfortunately, matrix transposition is trivial (zero cost!) in this model,
so while our reductions still make sense here,
they are completely useless.

If we want to stay closer to the Turing machine paradigm,
one natural variant to consider is the
\emph{multitape Turing machine with $d$-dimensional tapes}
(for fixed~$d$) \cite{Hen-online,Reg-superlinear}.
While transposition is not free in such a model,
it can easily be carried out in linear time (provided that $d \geq 2$),
suggesting that our reductions are not especially interesting here either.

\subsection{Other permutations beyond transposition}

Matrix transposition may be regarded as a specific type of
permutation of the input.
It would be interesting to know which other permutations can be
reduced to integer multiplication.  For instance, what about
the permutation $(x_0, \ldots, x_{2^\ell-1}) \mapsto
(x_{\rho(0)}, \ldots, x_{\rho(2^\ell-1)})$, where $\rho(i)$
stands for the bit-reversal of $i$ as an $\ell$-bit integer?

\subsection{Alternative coefficient rings}

In the main reduction (Theorem \ref{thm:main-medium}),
we chose to work over the complex numbers.
It may be possible to work instead with DFTs over a finite field $\FF_q$.
This approach has a long history in the
context of integer multiplication \cite{Pol-ntt},
and has the advantage of avoiding the analysis of numerical error.
However, it introduces new technical difficulties,
such as the construction of field extensions containing
suitable roots of unity.
We have not checked the details.

\subsection{Polynomial multiplication over finite fields}

Let $\Mcost_q(n)$ denote the cost (in the Turing model)
of multiplying polynomials in $\FF_q[x]$ of degree~$n$.
By analogy with the integer case,
it is widely suspected that $\Mcost_q(n) = \Theta(b \log b)$
where $b \coloneqq n \log q$ is the total bit size.
Currently, the best unconditional upper bound is
$\Mcost_q(n) = O(b \log b \cdot 4^{\lg^* b})$
\cite{HvdH-ffmul-cyclotomic},
and there is even a conditional upper bound $\Mcost_q(n) = O(b \log b)$
assuming an unproved but plausible
number-theoretic hypothesis \cite{HvdH-ffnlogn}.

It is natural to ask whether the methods of this paper
generalise to this situation, i.e.,
can we reduce matrix transposition to polynomial multiplication over $\FF_q$?
If so, then lower bounds on transposition
may imply lower bounds for $\Mcost_q(n)$.

A rough analogy for the Boolean circuit model in this context would be
some variant of algebraic complexity,
such as counting the number of arithmetic operations in $\FF_q$,
without regard for memory access or locality.
Again, any analogue of our reduction would be useless in such a model,
as transposition is presumably free of charge.

\subsection{Performing small products in parallel}

Consider the following problem:
given integers $u_1, \ldots, u_n$ and $v_1, \ldots, v_n$ of bit size $p$,
compute the products $u_1 v_1$, \ldots, $u_n v_n$.
Is there a way to reduce this problem to a single integer multiplication
problem of size $O(np)$?
(Or $O(1)$ such problems?)
We do not know how to do this,
but it is not implausible that such a reduction exists.
If this reduction were possible,
then many of the results of this paper could likely be improved,
and the proofs would simplify considerably.
Namely, in Theorem \ref{thm:main-medium},
the second term $l_1 l_2 n_1 n_2 \Mcost(b)$ could
probably be absorbed into the $\Mcost(m)$ term,
by performing all the multiplications by Bluestein factors
and twiddle factors ``in parallel''.
This would likely imply that many of the $\Mcostinc(m)$ terms
in Section \ref{sec:proofs} could be replaced by simply $\Mcost(m)$,
which would in turn allow us to strengthen the statements of
Theorem \ref{thm:main-ell-implication}
and Theorem \ref{thm:main-recursive-implication}
into more direct lower bounds for $\Mcost(m)$.

\subsection{Small versus large coefficients}

Our methods are weaker than ideal when applied to matrices
with small coefficients, such as the $b = 1$ case
(compare Theorem \ref{thm:main-recursive-implication} with
Theorem \ref{thm:log-implication}).
On the other hand,
one feels intuitively that transposing a matrix with $b$-bit entries
should be $b$ times more expensive than transposing a matrix
of the same dimensions with $1$-bit entries.
In symbols, we might reasonably expect that
\[
   |\Tcost(n_1, n_2, b) - b \Tcost(n_1, n_2, 1)| = O(n_1 n_2 b).
\]
Unfortunately, we have not yet succeeded in designing
reductions between transpositions with $b$-bit and
$1$-bit entries, in either direction.  If we could establish that
$\Tcost(n_1, n_2, b) = \Omega(b \Tcost(n_1, n_2, 1))$,
then we could immediately improve our results for small coefficients
(such as Theorem \ref{thm:main-recursive-implication}) by leveraging the
known results for larger coefficients.

If we restrict attention to \emph{oblivious} Turing machines,
i.e., machines for which the sequence of tape movements
depends only on the size of the input,
then there is a straightforward reduction in one direction:
given a machine handling the $1$-bit case,
we can transpose a $b$-bit matrix by simulating $b$ copies
of the $1$-bit machine in parallel.
Hence, in this model $\Tcost(n_1, n_2, b) = O(b \Tcost(n_1, n_2, 1))$.
Unfortunately, this reduction goes in the wrong direction to be useful
in the way suggested in the previous paragraph.

\bibliographystyle{amsalpha}
\bibliography{transpose}

\end{document}